\documentclass{amsart}
\usepackage{amsmath,amssymb}
\usepackage[colorlinks,linkcolor=blue,citecolor=blue,urlcolor=blue]{hyperref}
\usepackage[super,compress]{cite}
\usepackage{tkz-graph}
 \usetikzlibrary{arrows,automata}
\usepackage{color}
 \definecolor{green}{rgb}{0.00,0.80,0.00} 
 \definecolor{purple}{rgb}{0.60,0.22,0.65} 
 \definecolor{yellow}{rgb}{1.00,0.90,0.00} 
\usepackage{url}
\newcommand{\bz}{\mathbb{Z}}
\newcommand{\zt}{\mathbb{Z}_2}

\newcommand{\br}{\mathbb{R}}

\newcommand{\Ione}{\mathbb{I}}
\newcommand{\sym}{S}
\newcommand{\sign}{\Sigma}

\newcommand{\GR}{\mathcal{GR}}
\newcommand{\SL}{\mbox{SL}}
\newcommand{\GL}{\mbox{GL}}
\newcommand{\CM}{C\!M}
\newcommand{\TM}{T\!M}
\newcommand{\VM}{V\!M}
\newcommand{\sk}{\vspace{10pt}}

\newtheorem{theorem}{Theorem}
\newtheorem{lemma}[theorem]{Lemma}

\let\MC=\multicolumn

\def\baselinestretch{1.1}   \parskip\medskipamount     \lineskip=0pt
\allowdisplaybreaks
 \font\rOpe=cmsy10                        
 \def\ktl{{\hbox{\rOpe\char'170}}}        
 \def\kbl{{\hbox{\rOpe\char'170}}}        
 \def\kcr{{\reflectbox{\rOpe\char'170}}}        
 \def\ktr{{\reflectbox{\rOpe\char'170}}}        
 \def\kbr{{\reflectbox{\rOpe\char'170}}}        
 \def\Border{\vbox{\hsize0pt
        \setlength{\unitlength}{1mm}
        \newcount\xco
        \newcount\yco
        \xco=-21
        \yco=12
        \begin{picture}(0,0)(-7.5,0)
        \put(\xco,\yco){$\ktl$}
        \advance\yco by-1
        {\loop
        \put(\xco,\yco){$\kcr$}
        \advance\yco by-2
        \ifnum\yco>-240
        \repeat
        \put(\xco,\yco){$\kbl$}}
        \xco=170
        \yco=12
        \put(\xco,\yco){$\ktr$}
        \advance\yco by-1
        {\loop
        \put(\xco,\yco){$\kcr$}
        \advance\yco by-2
        \ifnum\yco>-240
        \repeat
        \put(\xco,\yco){$\kbr$}}
        \put(-19.5,13){\scalebox{.71}{%
         University of Maryland Center for String and Particle  Theory \&\ Physics Department%
        |University of Maryland Center for String and Particle  Theory \&\ Physics Department}}
        \put(-19.5,-241.5){\scalebox{.63}{
         Howard University Department of Physics and Astronomy%
         |University of Central Florida Dept. of Physics%
        |Pepperdine University Natural Science Division%
        |Bard College Mathematics Program}}
        \end{picture}
        \par\vskip-8mm}}
\definecolor{UMred}{rgb}{.9,.05,.2}
\definecolor{HUblue}{rgb}{.0,.3,.7}
 \def\UMbanner{\vbox{\hsize0pt
        \setlength{\unitlength}{.4mm}
        \thicklines\color{UMred}
        \begin{picture}(0,0)(-30,-10)
        \put(165,16){\line(1,0){4}}
        \put(170,16){\line(1,0){4}}
        \put(180,16){\line(1,0){4}}
        \put(175,0){\line(1,0){4}}
        \put(180,0){\line(1,0){4}}
        \put(185,0){\line(1,0){4}}
        \put(169,0){\line(0,1){16}}
        \put(170,0){\line(0,1){16}}
        \put(179,0){\line(0,1){16}}
        \put(180,0){\line(0,1){16}}
        \put(184,0){\line(0,1){16}}
        \put(185,0){\line(0,1){16}}
        \put(169,16){\oval(8,32)[bl]}
        \put(170,16){\oval(8,32)[br]}
        \put(179,0){\oval(8,32)[tl]}
        \put(185,0){\oval(8,32)[tr]}
        \color{HUblue}
        \put(167.75,-2){\line(1,0){4}}
        \put(172.75,-2){\line(1,0){4}}
        \put(177.75,-2){\line(1,0){4}}
        \put(182.75,-2){\line(1,0){4}}
        \put(167.75,-2){\line(0,-1){16}}
        \put(171.75,-2){\line(0,-1){16}}
        \put(172.75,-2){\line(0,-1){16}}
        \put(176.75,-2){\line(0,-1){16}}
        \put(181.75,-2){\line(0,-1){16}}
        \put(182.75,-2){\line(0,-1){16}}
        \put(181.75,-2){\oval(8,32)[bl]}
        \put(182.75,-2){\oval(8,32)[br]}
        \put(167.75,-18){\line(1,0){4}}
        \put(172.75,-18){\line(1,0){4}}
        \end{picture}
        \par\vskip-6.5mm
        \thicklines}}

\definecolor{Red}    {rgb}{0.90,0.00,0.12} 
\definecolor{Blue}   {rgb}{0.00,0.00,1.00} 
\definecolor{Green}  {rgb}{0.10,0.70,0.10} 
\definecolor{Turque} {rgb}{0.00,0.65,0.85} 
\definecolor{Orange} {rgb}{1.00,0.50,0.15} 
\definecolor{Magenta}{rgb}{1.00,0.00,1.00} 
\definecolor{Gold}   {rgb}{1.00,0.75,0.25} 
\definecolor{Seaweed}{rgb}{0.01,0.24,0.09} 
\definecolor{Purple} {rgb}{0.50,0.25,0.55} 
\definecolor{Brown}  {rgb}{0.43,0.26,0.32} 
\definecolor{grey1}  {rgb}{0.20,0.20,0.20} 
\definecolor{grey2}  {rgb}{0.40,0.40,0.40} 
\definecolor{grey3}  {rgb}{0.60,0.60,0.60} 
\definecolor{grey4}  {rgb}{0.80,0.80,0.80} 
\definecolor{grey5}  {rgb}{0.90,0.90,0.90} 
\def\C#1#2{{\ifcase#1\or
             \color{Red}\or \color{Green}\or \color{Blue}\or\
              \color{Turque}\or \color{Orange}\or \color{Magenta}\or 
               \color{Gold}\or \color{Seaweed}\or \color{Purple}\or
                \color{Brown}\or\color{grey1}\or\color{grey2}\or
                 \color{grey3}\else\color{grey4}\fi#2}}

\definecolor{Slate} {rgb}{0.00,0.45,0.55}

\newdimen\parshift\parshift=\parindent
\catcode`@=11
 \long\def\@footnotetext#1{\insert\footins{\reset@font\footnotesize
           \interlinepenalty\interfootnotelinepenalty\splittopskip%
            \footnotesep\splitmaxdepth\dp\strutbox\floatingpenalty\@MM%
             \hsize\columnwidth\addtolength{\hsize}{-2\parindent}
              \@parboxrestore\protected@edef\@currentlabel%
              {\csname p@footnote\endcsname\@thefnmark}%
                \color@begingroup%
                 \@makefntext{\rule\z@\footnotesep\ignorespaces#1%
                  \@finalstrut\strutbox}%
                \color@endgroup}}
 \long\def\@makefntext#1{\hglue\parshift%
           \vbox{\noindent\baselineskip=11pt plus.5pt minus.5pt\hb@xt@0em{\hss\@makefnmark\kern1pt}#1}}
\catcode`@=12

\newskip\humongous \humongous=0pt plus 1000pt minus 1000pt

\newif\ifdtup

\makeatletter
\def\section{\@startsection{section}{1}{\z@}
        {3ex plus-1ex minus-.2ex}{1pt plus1pt}{\large\sf\bfseries\boldmath}}
\def\subsection{\@startsection{subsection}{2}{\z@}
         {1.5ex plus-1ex minus-.2ex}{0.01pt plus1pt}{\sf\slshape}}
\def\subsubsection{\@startsection{subsubsection}{3}{\z@}
          {1.5ex plus-1ex minus-.2ex}{0.01pt plus0.2pt}{\sf\boldmath}}
\def\paragraph{\@startsection{paragraph}{4}{\z@}
           {.75ex \@plus.5ex \@minus.2ex}{-2mm}{\sf\bfseries\boldmath}}
\makeatother

\allowdisplaybreaks

\begin{document}

\thispagestyle{empty}
\vbox{\Border\UMbanner}
 \noindent{\small
 \today\hfill{PP-016-005 %
 }}
\vspace*{5mm}
\begin{center}
{\LARGE\sf\bfseries
${N=4}$ and ${N=8}$ SUSY Quantum Mechanics\\ [3mm]
and Klein's Vierergruppe}
\\[12mm]
{\large\sf\bfseries 
S.\ James Gates, Jr.$^{a}$\footnote{gatess@wam.umd.edu},~
T.\, H\"{u}bsch$^{b, \, c}$\footnote{thubsch@howard.edu},~
K.\, Iga$^{d}$\footnote{kiga@pepperdine.edu},~
and\,
S. Mendez--Diez$^{e}$\footnote{smendezdiez@bard.edu}
}\\*[8mm]
\emph{
\centering
$^a$Center for String and Particle Theory, Dept.\ of Physics, \\[-2pt]
University of Maryland, College Park, MD 20472 
\\[10pt]
$^{b}$Dept.\ of Physics \&\ Astronomy, Howard University, Washington, DC 20059
\\[10pt]
$^c$Dept.\ of Physics, University of Central Florida, Orlando, FL 32816
\\[10pt] 
$^d$Natural Science Division, Pepperdine University, 24255 Pacific Coast Hwy,
Malibu, CA 90263
\\[10pt] 
$^e$Mathematics Program, Bard College, Annandale-on-Hudson, NY 12504     
}
 \\*[50mm]
{\sf\bfseries ABSTRACT}\\[4mm]
\parbox{142mm}{\parindent=2pc\indent\baselineskip=14pt plus1pt
Sets of signed permutation matrices satisfying the $\GR(4,4)$ algebra are shown 
to be, up to sign, left cosets of Klein's famous Vierergruppe.  In this way we verify 
the count done by computer in 2012, and set it in a more significant mathematical 
context.  A similar analysis works for $\GR(1,1)$, $\GR(2,2)$ and $\GR(8,8)$.
 }
 \end{center}
\vfill
\vfill

\clearpage

\setlength{\unitlength}{1pt}
\def\baselinestretch{1}   \parskip\medskipamount     \lineskip=0pt
\setlength{\topmargin}{0in}      
\setlength{\textheight}{9in}
\setlength{\oddsidemargin}{0in}
\setlength{\evensidemargin}{\oddsidemargin}
\setlength{\hsize}{7in}
\setlength{\textwidth}{\hsize}












\section{Introduction}
The $\GR(d,N)$ algebras have been used in the study of supersymmetry in one dimension ($0$ space, $1$ time).\cite{rGR1,rGR2,rGR3,remains,enuf}  The $\GR(d,N)$ algebra involves $N>0$ signed permutation matrices $L_1, \ldots, L_N$, each of size $d\times d$, such that
\begin{equation}
\begin{aligned}
L_IL_J{}^T+L_JL_I{}^T&=2\delta_{IJ}\,\Ione\\
L_I{}^TL_J+L_J{}^TL_I&=2\delta_{IJ}\,\Ione,\label{eqn:garden}
\end{aligned}
\end{equation}
where ${}^T$ denotes matrix transpose, $\delta_{IJ}$ is the Kronecker delta symbol, and $\Ione$ is the identity matrix.  For each $N$, there is a minimal $d$ where this is possible, according to the theory of Clifford algebras.\cite{rGR1,rGR2} The formula for this minimal $d$ is
\begin{equation}
d_{\min}=\begin{cases}
2^{\frac{N}{2}-1},&N\equiv  0 \bmod{8}\\
2^{\frac{N-1}{2}},&N\equiv 1,\,7 \bmod{8}\\
2^{\frac{N}{2}},&N\equiv 2,\,4,\,6 \bmod{8}\\
2^{\frac{N+1}{2}},&N\equiv 3,\,5\bmod{8}.
\end{cases}
\label{eqn:dmin}
\end{equation}
The cases where $N$ is a multiple of $4$ is of greater interest because these can arise from dimensional reduction from four dimensions, among other reasons. 

In 2012, an exhaustive computer search using {\sl\/Mathematica\/\textsuperscript{TM}} for the case $N=4$, $d=d_{\min}=4$ found 1536 such sets, which furthermore fit into $6$ types, according to the permutation matrices obtained by entrywise absolute values of the $L_I$ matrices.\cite{pizza4}

Subsequently, in 2013, Stephen Randall, an undergraduate student at the University of Maryland, did a similar {\sl\/Mathematica\/\textsuperscript{TM}} search for $N=8$ and $d=d_{\min}=8$, and though the long computation time prevented finding the total number of sets of $L_I$ matrices, he was able to find that there were 151{,}200 sets of permutation matrices.\cite{pizza8}

In this paper, we show that $N=4$, $d=4$ case can be understood in terms of Klein's famous {\em Vierergruppe}.  This not only generates the numbers $6$ and 1536 without the aid of a computer, but it also gives meaning to the six types in terms of the theory of cosets.  This approach, generalized slightly, applies as well to the case $N=8$, $d=8$.  This allows for a non-computer counting of the 151{,}200 sets of permutation matrices, and also provides an answer to the number of sets of $L_I$ matrices.

\subsection{Signed permutation matrices}
A signed permutation matrix is a square matrix where each column and each row has exactly one non-zero entry, and this entry is either $1$ or $-1$.  A permutation matrix is a signed permutation matrix where all non-zero entries are $1$.  As a linear transformation, a permutation matrix acts as a permutation on the standard basis elements in $\br^d$.  These, in turn, can be viewed as permutations on the numbers $\{1,\ldots,d\}$.

Every signed permutation matrix $L$ can be written as a matrix product $L=\sign |L|$ where $\sign$ is a diagonal matrix with only $\pm 1$ entries, and $|L|$ is the matrix obtained from taking the absolute value of each of the entries of $L$.  The matrix $|L|$ is a permutation matrix.  The $(i,i)$ entry of the diagonal matrix $\sign$ is the sign of the non-zero entry in the $i$th row of $L$. See Figure~\ref{fig:signedperm}.
\begin{figure}[ht]
\[
\left[\begin{array}{cccc}
0&1&0&0\\
0&0&0&-1\\
0&0&-1&0\\
1&0&0&0
\end{array}\right]
=
\left[\begin{array}{cccc}
1&0&0&0\\
0&-1&0&0\\
0&0&-1&0\\
0&0&0&1
\end{array}\right]
\left[\begin{array}{cccc}
0&1&0&0\\
0&0&0&1\\
0&0&1&0\\
1&0&0&0
\end{array}\right]
\]
\caption{Factoring a signed permutation matrix into a diagonal and a permutation matrix}
\label{fig:signedperm}
\end{figure}

Note that $\sign$ and $|L|$ are orthogonal matrices, and since orthogonal matrices form a group, $L=\sign |L|$ is orthogonal as well.  Also, if $L$ and $L'$ are signed permutation matrices, then $|LL'|=|L||L'|$.  This fact is related to the fact that the set of $\sign$ matrices is a normal subgroup of the set of signed permutation matrices.

\subsection{Permutation notation}
The following is a review of some basic facts about permutations, which also serves to establish notation.

A permutation is a bijection from the set $\{1,\ldots,d\}$ to itself.  One way to write such a function is to write the numbers $1,\ldots, d$ in order from left to right on one line, and then the values of the function are written underneath, so that:
\[ \sigma=\left\langle \begin{array}{ccccc}
1&2&3&4&5\\
3&4&5&2&1
\end{array}\right\rangle \]
is the permutation $\sigma$ so that $\sigma(1)=3$, $\sigma(2)=4$, $\sigma(3)=5$, $\sigma(4)=2$, and $\sigma(5)=1$.  This notation is called {\em two line notation}.   Since the first line is always the same, we can omit it:
\[ \sigma=\langle 3\, 4\, 5\, 2\, 1\rangle \]
This notation is called {\em one line notation}.  Note the use of angle brackets here instead of parentheses, to avoid confusion with disjoint cycle notation, described below.

We compose permutations as functions, right-to-left, so that if $\sigma$ sends $1$ to $2$, and $\tau$ sends $2$ to $3$, then $\tau\sigma$ sends $1$ to $3$.  The set of permutations on $d$ elements then forms a group with $d!$ elements, called the symmetric group, and denoted $\sym_d$.

If $a_1, a_2, \ldots, a_n$ are a finite ordered sequence of different numbers from the set $\{1,\ldots,d\}$, then we can define a permutation called a cycle, denoted
\[ (a_1\, a_2\, \ldots\, a_n) \]
that sends $a_1$ to $a_2$, $a_2$ to $a_3$, and so on, up to $a_{n-1}$ to $a_n$; and also sends $a_n$ to $a_1$.  All other elements of $\{1,\ldots,d\}$, if they exist, are sent to themselves under this permutation.  Sometimes commas are included between the $a_i$ when omitting them might cause confusion.
 The identity is written $(\,)$.

Every permutation can be written as a product of disjoint cycles, for instance, with the example $\sigma$ above, we can write:
\[ \sigma=(1\,3\,5)(2\,4) \]
is a permutation which sends $1$ to $3$, sends $3$ to $5$, sends $5$ to $1$, sends $2$ to $4$, and sends $4$ to $2$.  We will typically write permutations in this {\em disjoint cycle notation} in this paper, though note that Ref.~\citen{pizza4} used the one line notation.

\subsection{$\GR(d,N)$ algebras and one dimensional supersymmetry}
The $N$-ex\-ten\-ded supersymmetry algebra in one dimension involves some surprisingly rich and beautiful mathematics.  The algebra is generated by one time translation operator $H=i\frac{d}{dt}$, and $N$ odd real operators $Q_1,\ldots, Q_N$, each commuting with $H$, with
\begin{equation}
\{Q_I,Q_J\}=2i\,\frac{d}{dt}.
\label{eqn:susy}
\end{equation}

The {\em isoscalar supermultiplet}, originally called the {\em Spinning Particle} in Refs.~\citen{rGR1,rGR2,rGR3} and the {\em scalar multiplet} in Ref.~\citen{rA}, is defined using $d$ bosonic fields $\Phi_i$ and $d$ fermionic fields $\Psi_i$, with transformation rules
\begin{equation}
\begin{aligned}
Q_I \Phi_i&=(L_I)_i{}^j\Psi_j\\
Q_I \Psi_i&=i(R_I)_i{}^j\frac{d}{dt}\Phi_j.
\end{aligned}
\end{equation}
The $L_1,\ldots,L_N$ and $R_1,\ldots,R_N$ are $d\times d$ signed permutation matrices.  The supersymmetry algebra~(\ref{eqn:susy}) then becomes
\begin{equation}
\begin{aligned}
L_IR_J+L_JR_I&=2\delta_{IJ}\Ione\\
R_IL_J+R_JL_I&=2\delta_{IJ}\Ione .\label{eqn:garden0}
\end{aligned}
\end{equation}
By setting $I=J$ it is easy to see $R_I=L_I{}^{-1}$, and since signed permutations are orthogonal, $R_I=L_I{}^T$.  Then~(\ref{eqn:garden0}) takes on the form~(\ref{eqn:garden}) above.

Thus, for every choice of $L_1,\ldots,L_N$ matrices satisfying these conditions, there is an off-shell one-dimensional isoscalar supermultiplet.  All bosons are of one engineering degree, and all fermions are of an engineering degree $1/2$ higher.

\subsection{Adinkras}
\label{sec:adinkra}
Graphs known as Adinkras were proposed by Faux and Gates in Ref.~\citen{rA} and precisely codified in Ref.~\citen{r6-1} as a fruitful way to investigate off-shell supermultiplets in one dimension.  In particular, they can be used to study these isoscalar multiplets.

An Adinkra is a graph, where vertices are drawn either with white or black dots, and edges are colored from among $N$ colors, and edges can either be solid or dashed.

In this context, given a set of $\{L_I\}$ matrices, the corresponding Adinkra has $d$ white vertices labeled $\Phi_1,\ldots,\Phi_d$ and $d$ black vertices labeled $\Psi_1,\ldots,\Psi_d$.  The white vertices are arranged from left to right in a horizontal line in numerical order, and the black vertices are likewise placed above them.  If $(L_I)_i{}^j$ is non-zero, an edge of color $I$ is drawn from vertex $\Phi_i$ to vertex $\Psi_j$, and is solid if $(L_I)_i{}^j$ is $1$, and dashed if it is -1.  See Figure~\ref{fig:adinkra4}.

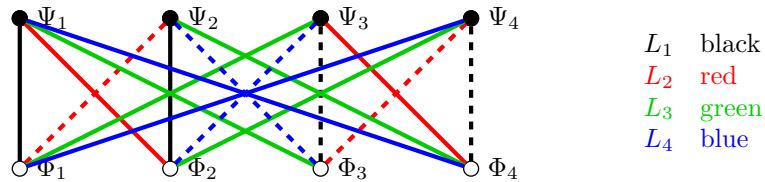
\begin{figure}[ht]
\begin{center}
\begin{tikzpicture}
\GraphInit[vstyle=Welsh]
\SetVertexNormal[MinSize=5pt]
\SetUpEdge[labelstyle={draw},style={ultra thick}]
\tikzset{Dash/.style = {ultra thick, dashed}}
\Vertex[x=0,y=0,Math,L={\Phi_1}]{B1}
\Vertex[x=2,y=0,Math,L={\Phi_2}]{B2}
\Vertex[x=4,y=0,Math,L={\Phi_3}]{B3}
\Vertex[x=6,y=0,Math,L={\Phi_4}]{B4}
\Vertex[x=0,y=2,Math,L={\Psi_1}]{F1}
\Vertex[x=2,y=2,Math,L={\Psi_2}]{F2}
\Vertex[x=4,y=2,Math,L={\Psi_3}]{F3}
\Vertex[x=6,y=2,Math,L={\Psi_4}]{F4}
\AddVertexColor{black}{F1,F2,F3,F4}
\Edge(B1)(F1)
\Edge(B2)(F2)
\Edge[style=Dash](B3)(F3)
\Edge[style=Dash](B4)(F4)
\Edge[color=red,style=Dash](B1)(F2)
\Edge[color=red](B2)(F1)
\Edge[color=red,style=Dash](B3)(F4)
\Edge[color=red](B4)(F3)
\Edge[color=green](B1)(F3)
\Edge[color=green](B2)(F4)
\Edge[color=green](B3)(F1)
\Edge[color=green](B4)(F2)
\Edge[color=blue](B1)(F4)
\Edge[color=blue,style=Dash](B2)(F3)
\Edge[color=blue,style=Dash](B3)(F2)
\Edge[color=blue](B4)(F1)
\node [right] at (8,1) {
\begin{tabular}{rl}
$L_1$&black\\
{\color{red} $L_2$}&{\color{red} red}\\
{\color{green} $L_3$}&{\color{green} green}\\
{\color{blue} $L_4$}&{\color{blue} blue}
\end{tabular}
};
\end{tikzpicture}
\end{center}
\caption{An Adinkra with $N=4$ and $d=4$, illustrating the $L_I$ matrices given in (\ref{eqn:livm3}) in Section~\ref{sec:vm3}.  The lower four vertices are bosons, and the top four vertices are fermions.  The colors correspond to each of the $L_I$ according to the table on the right.}
\label{fig:adinkra4}
\end{figure}

Adinkras in general may have bosons and fermions at many different levels, recording the engineering dimension of each vertex, but in our case, these are kept at two levels.  Such two-level Adinkras are called {\em valise} Adinkras.

Adinkras have been classified in a series of works.\cite{r6-1,r6-2,r6-codes,at,r6-8,cc,rPT,rT01,rKRT,rKT07}
The first step in this classification is the Hamming cube, which is the 1-skeleton of the $N$-dimensional cube $[0,1]^N$.  As a graph, the vertices are $N$-tuples of $0$'s and $1$'s, and edges connect vertices that differ in precisely one coordinate.  We color each such edge by color $I$ if the $I$th coordinate is the one that differs.  A vertex is a boson if the number of $1$'s is even, and a fermion otherwise.  

According to the classification of Adinkras, each Adinkra is a disjoint union of quotients of $N$-dimensional Hamming cubes by doubly even codes of length $N$.\cite{at}  A doubly even code of length $N$ is a subgroup of $\zt{}^N$ so that the number of $1$s in each element is a multiple of $4$.

One caution, however: in this classification, the graphs are abstract, and are considered equivalent up to isomorphism.  But for classifying sets of $L_I$ matrices, we must number the bosons and fermions $1$ through $d$.  In particular, if we renumber the vertices, we may end up with the same Adinkra, but different sets of $L_I$ matrices.  A discussion on the meaning of this difference is relegated to Section~\ref{sec:metaphysics}.

\section{$\VM_3$ and the Vierergruppe}
\label{sec:vm3}
As described in the Introduction, an exhaustive computer search found six sets of permutation matrices that come from $L_I$ matrices with $N=4$, $d=4$.\cite{pizza4}  The sixth such set, called $\VM_3$, is our starting point.  The set of $L_I$ matrices is as follows (this corresponds to the Adinkra in Figure~\ref{fig:adinkra4}).
\begin{equation}
\begin{aligned}
L^{\VM_3}_1&=\left[\begin{matrix}
1&0&0&0\\
0&1&0&0\\
0&0&-1&0\\
0&0&0&-1
\end{matrix}\right],
&\qquad
L^{\VM_3}_2&=\left[\begin{matrix}
0&-1&0&0\\
1&0&0&0\\
0&0&0&-1\\
0&0&1&0
\end{matrix}\right],\\
L^{\VM_3}_3&=\left[\begin{matrix}
0&0&1&0\\
0&0&0&1\\
1&0&0&0\\
0&1&0&0
\end{matrix}\right],
&\qquad
L^{\VM_3}_4&=\left[\begin{matrix}
0&0&0&1\\
0&0&-1&0\\
0&-1&0&0\\
1&0&0&0
\end{matrix}\right].
\end{aligned}
\label{eqn:livm3}
\end{equation}

If we drop the signs, the result is a set of permutation matrices:
\begin{equation*}
\begin{aligned}
v_1=|L^{\VM_3}_1|&=\left[\begin{matrix}
1&0&0&0\\
0&1&0&0\\
0&0&1&0\\
0&0&0&1
\end{matrix}\right],
&\qquad
v_2=|L^{\VM_3}_2|&=\left[\begin{matrix}
0&1&0&0\\
1&0&0&0\\
0&0&0&1\\
0&0&1&0
\end{matrix}\right],\\
v_3=|L^{\VM_3}_3|&=\left[\begin{matrix}
0&0&1&0\\
0&0&0&1\\
1&0&0&0\\
0&1&0&0
\end{matrix}\right],
&\qquad
v_4=|L^{\VM_3}_4|&=\left[\begin{matrix}
0&0&0&1\\
0&0&1&0\\
0&1&0&0\\
1&0&0&0
\end{matrix}\right].
\end{aligned}
\end{equation*}
In disjoint cycle notation these permutations become:
\begin{equation*}
\begin{aligned}
v_1&=(\,),&\qquad
v_2&=(1\,2)(3\,4),\\
v_3&=(1\,3)(2\,4),&\qquad
v_4&=(1\,4)(2\,3).
\end{aligned}
\end{equation*}
One striking fact about this case is that $v_4=v_2v_3=v_3v_2$, $v_2=v_3v_4=v_4v_3$, and $v_3=v_2v_4=v_4v_2$, so that $V=\{v_1, v_2, v_3, v_4\}$ forms a subgroup of the symmetric group $\sym_4$, with the following multiplication table:
\begin{equation*}
\begin{array}{c|cccc}
&v_1&v_2&v_3&v_4\\\hline
v_1&v_1&v_2&v_3&v_4\\
v_2&v_2&v_1&v_4&v_3\\
v_3&v_3&v_4&v_1&v_2\\
v_4&v_4&v_3&v_2&v_1
\end{array}
\end{equation*}
This $V$ is abelian, and is isomorphic to $\bz_2\times\bz_2$.  It is a normal subgroup of $\sym_4$.  It is in fact the famous Klein Vierergruppe popularized by F. Klein in his 1884 text, {\em Lectures on the Icosahedron and Equations of the Fifth Degree}.\cite{klein} Among other possible descriptions, the Vierergruppe can be thought of as acting on a square by horizontal reflection, vertical reflection, and rotation by $180$ degrees around the center (which is the composition of both reflections), as shown in Figure~\ref{fig:square}.  An alternative point of view is that the Vierergruppe is the group of translations modulo $2$: add an ordered pair of integers $(a,b)$ and reduce both coordinates modulo $2$.  This can be put in a physical context if we imagine that the square is a part of an infinite two dimensional square lattice, which is periodic in both the horizontal and vertical directions by translations, and the Vierergruppe is what happens when we translate by half a period.
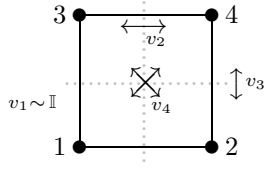
\begin{figure}[ht]
\begin{center}
  \begin{picture}(50,60)
   \color{lightgray}
   \multiput(23,-5.75)(0,3){21}{.}
   \multiput(-6,23.5)(3,0){21}{.}
   \color{black}
   \put(-27,15){$\scriptstyle v_1\sim\,\Ione$}
   \put(56,21.5){$\big\updownarrow\rlap{$\vcenter{\hbox{$\scriptstyle v_3$}}$}$}
   \put(25,39){$\scriptstyle v_2$}
   \put(15,43){$\longleftrightarrow$}
   \put(21,23){$\nearrow$}
   \put(19,23){$\nwarrow$}
   \put(21,21){$\searrow$}
   \put(19,21){$\swarrow$}
   \put(27,14){$\scriptstyle v_4$}
   \put(0,0){\line(1,0){50}}
   \put(0,0){\line(0,1){50}}
   \put(50,50){\line(-1,0){50}}
   \put(50,50){\line(0,-1){50}}
   \put(0,0){\circle*{5}}   \put(-10,-3){1}
   \put(50,0){\circle*{5}}  \put(55,-3){2}
   \put(0,50){\circle*{5}}  \put(-10,47){3}
   \put(50,50){\circle*{5}} \put(55,47){4}
  \end{picture}
\end{center}
\caption{A square, regarded as a two-dimensional affine space over the field $\zt$, i.e., the fundamental domain in a square lattice.  The Vierergruppe $V$ may then be regarded as the group of translations modulo $2$, or equivalently of the indicated reflections.}
\label{fig:square}
\end{figure}

\subsection{$\GR(4,4)$ and the Vierergruppe}
\label{sec:n4}
We now consider more general isoscalar supermultiplets with $N=4$, $d=4$, from the Adinkra perspective.  The Hamming cube $\{0,1\}^4$ has $2^4=16$ vertices, so $d=8$ bosons and $8$ fermions.  In order to get the minimal $d=4$, we need to quotient by a doubly even code.  There is only one non-trivial doubly even code:\cite{rCHVP}
\[ d_4=\{(0,0,0,0),\,(1,1,1,1)\}. \]
It is precisely by quotienting the Hamming cube by this code that we have an Adinkra with $d=4$ bosons and $d=4$ fermions.

Therefore all the valise Adinkras of $\GR(4,4)$ look the same, if we ignore the dashings.  The only issue, then, is identifying which boson is $\Phi_1$, which one is $\Phi_2$, and so on, and which fermion is $\Psi_1$, which one is $\Psi_2$, and so on.  The choice for $\VM_3$ is one such choice.  The choice for the others involves a different ordering of the bosons and fermions, which, relative to the choice for $\VM_3$, is a permutation $\sigma$ of the bosons and a permutation $\tau$ of the fermions.

Thus, if $\{L_1,L_2,L_3,L_4\}$ satisfy the $\GR(4,4)$ equations in~(\ref{eqn:garden}), then there are permutations $\sigma$ and $\tau$ so that $\{|L_1|, \ldots, |L_4|\}=\{\sigma v_1\tau , \ldots, \sigma v_4\tau\}$.  In the language of cosets, there are permutations $\sigma$ and $\tau$ so that $\{|L_1|,\ldots,|L_4|\}=\sigma V\tau$.

Since $V$ is a normal subgroup of $\sym_4$, we have $\sigma V\tau = \sigma\tau V=\pi V$, where we define $\pi=\sigma\tau$.  We have thus proved:
\begin{theorem}
For $\GR(4,4)$, sets of permutations $\{|L_1|,\ldots,|L_4|\}$ are left cosets of the Vierergruppe $V$ in $\sym_4$.
\label{thm:coset4}
\end{theorem}

The set of left cosets of $V$ is $\sym_4/V\cong \sym_3$.  Therefore there are six choices for $\{|L_1|,\ldots, |L_4|\}$, one for each element of $\sym_3$:
\begin{equation}
\begin{array}{l@{\>=\>}r@{\>=\>}l}
  \VM_3 &  V          &\{(\,), (1\,2)(3\,4), (1\,3)(2\,4), (1\,4)(2\,3)\},\\
  \VM_2 &  (1\,2)V    &\{(1\,2), (3\,4), (1\,3\,2\,4), (1\,4\,2\,3)\},\\
  \VM_1 &  (1\,3)V    &\{(1\,3), (1\,2\,3\,4), (2\,4), (1\,4\,3\,2)\},\\
  \VM   &  (2\,3)V    &\{(2\,3), (1\,3\,4\,2), (1\,2\,4\,3), (1\,4)\},\\
  \CM   &\>(1\,2\,3)V &\{(1\,2\,3), (1\,3\,4), (2\,4\,3), (1\,4\,2)\}\\
  \TM   &\>(1\,3\,2)V &\{(1\,3\,2), (2\,3\,4), (1\,2\,4), (1\,4\,3)\}.
\end{array} 
\end{equation}
Here we have added the notation for these six types from Ref.~\citen{pizza4}: $\CM$, $\VM$, and $\TM$ are dimensional reductions of the $4d$, $N=1$ chiral, vector, and tensor supermultiplets, respectively, and $\VM_1$, $\VM_2$, and $\VM_3$ are the other types.

Thus we have six choices for the $\{|L_1|, \ldots, |L_4|\}$.  Each of these six choices correspond to many sets of signed permutations $\{L_1,\ldots,L_4\}$, as we will explain in the following section.

\subsection{Choosing signs}
To find how many sets $\{L_1,\ldots,L_N\}$ correspond to a set of permutations $\{|L_1|,\ldots,|L_N|\}$, we have the following Lemma:

\begin{lemma}
If the $\{L_1,\ldots,L_N\}$ satisfy the Garden Algebra~(\ref{eqn:garden}), and describe a connected Adinkra, with code $C$.  Then there are
\begin{equation}
2^{2d-1}\cdot |C|
\label{eqn:signs}
\end{equation}
choices for $\sign_1,\ldots,\sign_N$ so that $\{\sign_1L_1,\ldots,\sign_NL_N\}$ also satisfy the Garden Algebra.  Here, $|C|$ means the number of elements of the code $C$.
\label{lem:signs}
\end{lemma}
\begin{proof}
Suppose $\{L_1,\ldots,L_N\}$ is a set of $d\times d$ signed permutation matrices satisfying~(\ref{eqn:garden}).  Now, given a boson or a fermion, we can replace it with its negative.  On the level of $L_I$ matrices, this means multiplying the corresponding row (resp.\ column) of all of the $L_I$ matrices by $-1$.  On the level of Adinkras, this takes all solid edges incident to the corresponding vertex and replaces them with dashed edges, and replaces all dashed edges incident to that vertex and replaces them with solid edges.  This possibility was explained in Ref.~\citen{rA}, and was formalized in Refs.~\citen{zhang,cc,rDGW-AA}, where it was called {\em vertex switching}.  By composing these, we get $2^{2d}$ vertex switches, and, assuming the Adinkra is connected, we get the same dashing if and only if it is all of the bosons and all of the fermions that are vertex switched.  Therefore there are $2^{2d-1}$ sets of $\{L_1,\ldots,L_N\}$ that are obtained in this way.

The choices of signs for the $\{L_I\}$ matrices modulo vertex switching is in bijection with the code itself.\cite{cc,zhang}  The lemma follows.
\end{proof}

In the case of $\GR(4,4)$, there are
\[ 2^7\cdot 2 = 256 \]
possibilities for the sign for each of the six sets of permutation matrices.  Overall, then, there are 
\[ 6\cdot 256=1536 \]
choices for the $\{L_1,\ldots,L_4\}$ matrices satisfying the $\GR(d,N)$ algebra.  This was a fact first computed by exhaustive computer search in Ref.~\citen{pizza4}, but now we see the mathematical structure that gives rise to this fact.

Our approach followed the following steps.
\begin{itemize}
\item Deal with the permutations.
\begin{itemize}
\item Choose a canonical set $\{L_1, \ldots, L_4\}$ of matrices satisfying the $\GR(d,N)$ algebra, and take the corresponding set of permutations.  A good choice is available: $\VM_3$, whose set of permutations is $V$, the Vierergruppe.
\item Since $V$ is a subgroup, observe that every numbered $N=4$ valise Adinkra is related to this one via double cosets.
\item Since $V$ is normal, observe that these double cosets can be described as left cosets.
\item Find all left cosets of $V$, which is $\sym_4/V$.
\end{itemize}
\item Multiply by~(\ref{eqn:signs}) in Lemma~\ref{lem:signs} to find the total number of sets of $\{L_1,\ldots,L_4\}$ matrices.
\end{itemize}

\subsection{Duality}
One of the points of interest in Ref.~\citen{pizza4} was a certain duality operation, which sends each $L_I$ to $R_I=L_I{}^T=L_I{}^{-1}$.
Some collections $\{L_1,\ldots,L_N\}$ are preserved under this operation; that is,
\[ \{L_1{}^{-1},\ldots,L_4{}^{-1}\}=\{L_1,\ldots,L_4\}. \]
These are called {\em self-dual}.  Note that under this operation the $L_I$ are not necessarily equal to their inverses; all we require is that the inverse of each of the $L_I$ is another (possibly the same) $L_I$ within the same collection.

In order for a set $\{L_1,\ldots,L_4\}$ to be self-dual, it is necessary that their corresponding permutation set $\{|L_1|,\ldots,|L_4|\}$ be self-dual.  That is,
\[ \{|L_1|{}^{-1},\ldots,|L_4|{}^{-1}\}=\{|L_1|,\ldots,|L_4|\}. \]

The permutation sets that are self-dual were found in Ref.~\citen{pizza4} by direct computation.  Here we see that the coset notation can reproduce these findings easily.  As before, we write
\[ \sigma V = \{|L_1|,\ldots,|L_4|\} \]
and then we require $(\sigma V)^{-1}=\sigma V$.

Now $(\sigma V)^{-1}=V^{-1}\sigma^{-1}=V\sigma^{-1}=\sigma^{-1}V$, where we have used the fact that $V$ is normal.  Therefore, $\sigma V$ is self-dual if and only if $\sigma V$ is an element of order 1 or 2 in $\sym_4/V\cong \sym_3$.

Hence the following are self-dual:
\begin{equation}
\begin{array}{l@{\>=\>}r@{\>=\>}l}
\VM_3 &       V &\{(\,),(1\,2)(3\,4), (1\,3)(2\,4), (1\,4)(2\,3)\},\\
\VM_2 & (1\,2)V &\{(1\,2), (3\,4), (1\,3\,2\,4), (1\,4\,2\,3)\},\\
\VM_1 & (1\,3)V &\{(1\,3), (1\,2\,3\,4), (2\,4), (1\,4\,3\,2)\},\\
\VM   & (2\,3)V &\{(2\,3), (1\,3\,4\,2), (1\,2\,4\,3), (1\,4)\}.
\end{array}
\end{equation}
This duality also swaps the following:
\begin{equation}
\begin{array}{l@{\>=\>}r@{\>=\>}l}
 \CM &(1\,2\,3)V &\{(1\,2\,3), (1\,3\,4), (2\,4\,3), (1\,4\,2)\},\\
 \TM &(1\,3\,2)V &\{(1\,3\,2), (2\,3\,4), (1\,2\,4), (1\,4\,3)\}.
\end{array}
\end{equation}
This reproduces the result in Ref.~\citen{pizza4}, together with the fact that in $\VM_3$, each element is itself self-dual.

In turn, in each of $\VM,\VM_1$ and $\VM_2$, two of the elements are of order-2 and self-dual, while duality exchanges the remaining two (order-4) elements. For example,
\begin{equation}
 \begin{aligned}
  \VM_1 &=\{(1\,3), (1\,2\,3\,4), (2\,4), (1\,4\,3\,2)\}\\
        &\mapsto\{(1\,3), (1\,4\,3\,1), (2\,4), (1\,1\,3\,4)\}
\end{aligned}
 ~\big\updownarrow\rlap{$\vcenter{\hbox{$\scriptstyle L_2\leftrightarrow L_4$}}$}
\end{equation}
so that in $\VM_1$, duality effectively swaps the action of the corresponding supersymmetries. Since the same four supersymmetries thereby act differently on $\VM_1$ and on its dual, using both of these types of supermultiplets jointly in a model will provide for a richer dynamics than using only one of the two types; they are ``usefully distinct''~\cite{rChiLin}, not unlike the better known example of chiral and twisted chiral superfields in 1+1-dimensional spacetime~\cite{rGHR,rTwSJG0}.

\section{$N=8$}
\label{sec:n8}
Many of the ideas from $N=4$, $d=4$ also work for $N=8$, $d=8$.  The reader can check that for $N=8$, the minimal number for $d$ (the number of bosons), according to~(\ref{eqn:dmin}), is $d=8$.  The Hamming cube $\{0,1\}^8$ has $2^8=256$ vertices, so $128$ bosons and $128$ fermions, but by quotienting by a code of size $2^4=16$, we will get $d=8$ bosons and $d=8$ fermions.  Indeed, the unique\cite{rCHVP} maximal doubly even code for $N=8$ is $e_8$, which is generated by the rows of the matrix:\footnote{The columns of this code have been permuted from the description of $e_8$ on p. 373 of Ref.~\citen{rCHVP} for unimportant, but notationally convenient, reasons.  This is a description that is closer to the version found on p. 6 of the same reference, where it is called $\widehat{\mathcal{H}}_3$, the Hamming $(8,4)$ code.  The notation $e_8$ relates to the lattice $E_8$, which, by Construction A, involves taking the integer $8$-tuples whose reduction modulo $2$ is in the code.  This, in turn, relates to the Lie group $E_8$ whose roots are the nonzero lattice points closest to the origin.\cite{rCS}}
\[ e_8=\left[\begin{array}{cccccccc}
0&1&1&1&1&0&0&0\\
1&0&1&1&0&1&0&0\\
1&1&0&1&0&0&1&0\\
1&1&1&0&0&0&0&1
\end{array}\right] \]
The code consists of all linear combinations of these four rows, so has $2^4=16$ elements.

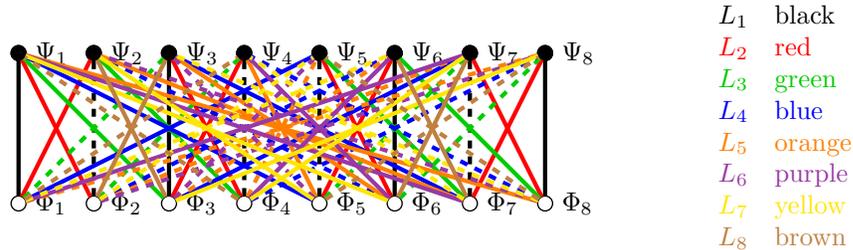
\begin{figure}[ht]
\begin{center}
\begin{tikzpicture}
\GraphInit[vstyle=Welsh]
\SetVertexNormal[MinSize=5pt]
\SetUpEdge[labelstyle={draw},style={ultra thick}]
\tikzset{Dash/.style = {ultra thick, dashed}}
\Vertex[x=0,y=0,Math,L={\Phi_1}]{B1}
\Vertex[x=1,y=0,Math,L={\Phi_2}]{B2}
\Vertex[x=2,y=0,Math,L={\Phi_3}]{B3}
\Vertex[x=3,y=0,Math,L={\Phi_4}]{B4}
\Vertex[x=4,y=0,Math,L={\Phi_5}]{B5}
\Vertex[x=5,y=0,Math,L={\Phi_6}]{B6}
\Vertex[x=6,y=0,Math,L={\Phi_7}]{B7}
\Vertex[x=7,y=0,Math,L={\Phi_8}]{B8}
\Vertex[x=0,y=2,Math,L={\Psi_1}]{F1}
\Vertex[x=1,y=2,Math,L={\Psi_2}]{F2}
\Vertex[x=2,y=2,Math,L={\Psi_3}]{F3}
\Vertex[x=3,y=2,Math,L={\Psi_4}]{F4}
\Vertex[x=4,y=2,Math,L={\Psi_5}]{F5}
\Vertex[x=5,y=2,Math,L={\Psi_6}]{F6}
\Vertex[x=6,y=2,Math,L={\Psi_7}]{F7}
\Vertex[x=7,y=2,Math,L={\Psi_8}]{F8}
\AddVertexColor{black}{F1,F2,F3,F4,F5,F6,F7,F8}
\Edge[color=black](B1)(F1)
\Edge[color=black,style=Dash](B2)(F2)
\Edge[color=black](B3)(F3)
\Edge[color=black,style=Dash](B4)(F4)
\Edge[color=black,style=Dash](B5)(F5)
\Edge[color=black](B6)(F6)
\Edge[color=black,style=Dash](B7)(F7)
\Edge[color=black](B8)(F8)
\Edge[color=red](B1)(F2)
\Edge[color=red](B2)(F1)
\Edge[color=red](B3)(F4)
\Edge[color=red](B4)(F3)
\Edge[color=red](B5)(F6)
\Edge[color=red](B6)(F5)
\Edge[color=red](B7)(F8)
\Edge[color=red](B8)(F7)
\Edge[color=green,style=Dash](B1)(F3)
\Edge[color=green,style=Dash](B2)(F4)
\Edge[color=green](B3)(F1)
\Edge[color=green](B4)(F2)
\Edge[color=green,style=Dash](B5)(F7)
\Edge[color=green,style=Dash](B6)(F8)
\Edge[color=green](B7)(F5)
\Edge[color=green](B8)(F6)
\Edge[color=blue](B1)(F5)
\Edge[color=blue,style=Dash](B2)(F6)
\Edge[color=blue](B3)(F7)
\Edge[color=blue,style=Dash](B4)(F8)
\Edge[color=blue](B5)(F1)
\Edge[color=blue,style=Dash](B6)(F2)
\Edge[color=blue](B7)(F3)
\Edge[color=blue,style=Dash](B8)(F4)
\Edge[color=orange,style=Dash](B1)(F8)
\Edge[color=orange,style=Dash](B2)(F7)
\Edge[color=orange,style=Dash](B3)(F6)
\Edge[color=orange,style=Dash](B4)(F5)
\Edge[color=orange](B5)(F4)
\Edge[color=orange](B6)(F3)
\Edge[color=orange](B7)(F2)
\Edge[color=orange](B8)(F1)
\Edge[color=purple](B1)(F7)
\Edge[color=purple,style=Dash](B2)(F8)
\Edge[color=purple,style=Dash](B3)(F5)
\Edge[color=purple](B4)(F6)
\Edge[color=purple,style=Dash](B5)(F3)
\Edge[color=purple](B6)(F4)
\Edge[color=purple](B7)(F1)
\Edge[color=purple,style=Dash](B8)(F2)
\Edge[color=yellow,style=Dash](B1)(F6)
\Edge[color=yellow,style=Dash](B2)(F5)
\Edge[color=yellow](B3)(F8)
\Edge[color=yellow](B4)(F7)
\Edge[color=yellow](B5)(F2)
\Edge[color=yellow](B6)(F1)
\Edge[color=yellow,style=Dash](B7)(F4)
\Edge[color=yellow,style=Dash](B8)(F3)
\Edge[color=brown,style=Dash](B1)(F4)
\Edge[color=brown](B2)(F3)
\Edge[color=brown](B3)(F2)
\Edge[color=brown,style=Dash](B4)(F1)
\Edge[color=brown,style=Dash](B5)(F8)
\Edge[color=brown](B6)(F7)
\Edge[color=brown](B7)(F6)
\Edge[color=brown,style=Dash](B8)(F5)
\node [right] at (9,1) {
\begin{tabular}{rl}
$L_1$&black\\
{\color{red} $L_2$}&{\color{red} red}\\
{\color{green} $L_3$}&{\color{green} green}\\
{\color{blue} $L_4$}&{\color{blue} blue}\\
{\color{orange} $L_5$}&{\color{orange} orange}\\
{\color{purple} $L_6$}&{\color{purple} purple}\\
{\color{yellow} $L_7$}&{\color{yellow} yellow}\\
{\color{brown} $L_8$}&{\color{brown} brown}
\end{tabular}
};
\end{tikzpicture}
\end{center}
\caption{An Adinkra with $N=8$ and $d=8$.}
\label{fig:adinkra8}
\end{figure}

An example of an Adinkra with $N=8$, $d=8$ is shown in Figure~\ref{fig:adinkra8}.  This provides the following permutations:
\begin{equation}
\begin{aligned}
a_1=|L_1|&=(\,),\\
a_2=|L_2|&=(1\,2)(3\,4)(5\,6)(7\,8),\\
a_3=|L_3|&=(1\,3)(2\,4)(5\,7)(6\,8),\\
a_4=|L_4|&=(1\,5)(2\,6)(3\,7)(4\,8),\\
a_5=|L_5|&=(1\,8)(2\,7)(3\,6)(4\,5),\\
a_6=|L_6|&=(1\,7)(2\,8)(3\,5)(4\,6),\\
a_7=|L_7|&=(1\,6)(2\,5)(3\,8)(4\,7),\\
a_8=|L_8|&=(1\,4)(2\,3)(5\,8)(6\,7).
\end{aligned}
 \label{e:A}
\end{equation}
Again, this set $A=\{a_1,\ldots,a_8\}$ is actually a group, and is isomorphic to $\bz_2\times\bz_2\times\bz_2$.  It is a subgroup of $\sym_8$, just as $V$ was a subgroup of $\sym_4$, but note that this time, $A$ is not a normal subgroup.

\begin{figure}[ht]
\begin{center}
\begin{picture}(100,100)(-10,-10)
\put(0,0){\circle*{5}}
\put(50,0){\circle*{5}}
\put(0,50){\circle*{5}}
\put(50,50){\circle*{5}}
\put(30,20){\circle*{5}}
\put(80,20){\circle*{5}}
\put(80,70){\circle*{5}}
\put(30,70){\circle*{5}}
\put(0,0){\line(1,0){50}}
\put(0,0){\line(0,1){50}}
\put(50,50){\line(-1,0){50}}
\put(50,50){\line(0,-1){50}}
\put(30,20){\line(1,0){17}}
\put(53,20){\line(1,0){27}}
\put(30,20){\line(0,1){27}}
\put(30,53){\line(0,1){17}}
\put(80,70){\line(-1,0){50}}
\put(80,70){\line(0,-1){50}}
\put(0,0){\line(3,2){30}}
\put(50,0){\line(3,2){30}}
\put(0,50){\line(3,2){30}}
\put(50,50){\line(3,2){30}}
\put(-10,-3){1}
\put(55,-8){2}
\put(-10,52){3}
\put(46,54){4}
\put(22,22){5}
\put(82,22){6}
\put(22,73){7}
\put(82,73){8}
\end{picture}
\end{center}
\caption{A cube.  This should be considered a three-dimensional affine space over the field $\zt$.  The group $A$ should be thought of as translations modulo $2$.}
\label{fig:cube}
\end{figure}
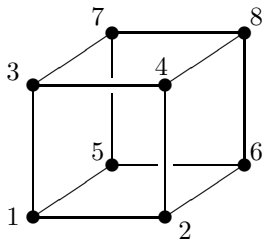

As with the Vierergruppe, these elements can be described in terms of the cube in Figure~\ref{fig:cube}.  More specifically, $a_2$, $a_3$, and $a_4$ are reflections in planes perpendicular to the faces, and the others are compositions of these.  An alternative interpretation is that the elements of $A$ are precisely the eight translations modulo $2$ of a three dimensional affine space defined over $\zt$, or if the reader prefers, translations of a periodic cubical lattice by half periods.

Comparing Figures~\ref{fig:square} and~\ref{fig:cube}, it is clear that the Vierergruppe has multiple images within the group $A$, corresponding to the multiple ways in which the $N=4$ supersymmetry algebra may be embedded in the $N=8$ algebra.

Again, since there is only one maximal doubly even code for $N=8$, every numbered valise corresponds to $A$ with a relabeling of the bosons and fermions.\cite{at}  If $\sigma$ and $\tau$ are the relabeling of the bosons and the fermions, respectively, then the $L_I$ matrices are of the form
\[ \sigma A \tau. \]
This can be written
\[ \sigma\tau ( \tau^{-1} A \tau) \]
or
\[ \pi B \]
where $\pi=\sigma\tau\in \sym_8$ and $B$ is a conjugate subgroup to $A$ in $\sym_8$.

Thus, we have:
\begin{theorem}
For $\GR(8,8)$, sets of permutations $\{|L_1|,\ldots,|L_8|\}$ are left cosets of conjugates of $A$ in $\sym_8$.
\end{theorem}

To count these sets of permutations, we need to count the conjugates of $A$ in $\sym_8$, and then count the left cosets of these conjugates.

To determine the set of conjugates to $A$, we let $\sym_8$ act on the set of conjugate subgroups of $A$ by conjugation.  The stabilizer of this action is $N(A)$, the normalizer of $A$ in $\sym_8$.

\begin{lemma}
The normalizer of $A$ in $\sym_8$ is the group of non-degenerate affine transformations
\begin{equation}
f(\vec{x})=L\vec{x}+\vec{t},
\label{eqn:affine}
\end{equation}
where $L$ is a $3\times 3$ non-degenerate matrix over $\bz_2$ and $\vec{t}$ is a translation vector in $\zt{}^3$.
\label{lem:normalizer} 
\end{lemma}

A $3\times 3$ non-degenerate matrix over $\zt$ is an element of $\GL_3(\zt)$, and since the determinant is an element of $\zt$ and is non-zero, it must be 1.  So $\GL_3(\zt)=\SL_3(\zt)$, which is the second smallest non-cyclic simple group, with order $168$.\footnote{This fact, and other details about this group can be found in Ref.~\citen{ATLAS} or online in Ref.~\citen{ATLASO}, where this group is called $L_3(2)$.}  In this way the set of affine transformations forms a semidirect product of $A$ and $\SL_3(\zt)$:  
\[ N(A)=A\rtimes \SL_3(\zt). \]

\begin{proof}
For a permutation $f:A\to A$ to be in the normalizer of $A$ means that for every $y\in A$, there is a $y'\in A$ so that for all $x\in A$,
\begin{equation}
f(f^{-1}(x)+y)=x+y'.
\label{eqn:conjugate}
\end{equation}
It is straightforward to see that functions $f$ of the form in~(\ref{eqn:affine}) satisfy this equation.

Conversely, suppose $f$ is a permutation that satisfies~(\ref{eqn:conjugate}).  Let $x_0=f^{-1}(\vec{0})$.  Define
\[ g(x)=f(x+x_0). \]
Then $g(\vec{0})=f(x_0)=\vec{0}$.  Computing the inverse of $g$ we get:
\[ g^{-1}(x)=f^{-1}(x)-x_0 \]
and then we see that
\begin{align*}
g(g^{-1}(x)+y),
&=f(f^{-1}(x)-x_0+y+x_0),\\
&=f(f^{-1}(x)+y),\\
&=x+y'.
\end{align*}
Therefore $g$ satisfies~(\ref{eqn:conjugate}) with $g(\vec{0})=\vec{0}$.

Letting $x=\vec{0}$ in this equation gives
\[ g(y)=y' \]
so we get
\[ g(g^{-1}(x)+y)=x+g(y) \]
and letting $u=g^{-1}(x)$ we get
\[ g(u+y)=g(u)+g(y) \]
so that $g$ is a group homomorphism.  For vector spaces over $\zt$, this is equivalent to being a linear transformation.  Therefore
\[ g(x)=Lx \]
and
\[ f(x)=g(x-x_0)=Lx-Lx_0 \]
which is an affine transformation.
\end{proof}

The group $A$ has $8$ elements, and the group $\SL_3(\zt)$ has $168$ elements,\cite{ATLAS,ATLASO} so the normalizer has $8\cdot 168=1344$ elements.  Therefore the number of conjugate subgroups of $A$ in $\sym_8$ is $8!/1344=30$.

For each conjugate subgroup, the number of left cosets is $8!/8=7!=5040$.  Overall, then, there are $5040\cdot 30 = 151{,}200$ such choices of $\{|L_1|,\ldots,|L_8|\}$ matrices.

The formula in Lemma~\ref{lem:signs} provides the number of sign matrices as
\[ 2^{16-1}\cdot 2^{4}=2^{19}=524{,}288. \]

Overall, then, the number of $\{L_1,\ldots,L_8\}$ sets is
\[ 151{,}200\cdot 524{,}288 = 79{,}272{,}345{,}600. \]

\begin{equation}
\begin{array}{r|r|r}
&\MC1{c|}{N=4}&\MC1c{N=8}\\\hline
\text{conjugates}&1&30\\
\text{left cosets each}&6&5040\\
\{|L_1|,\ldots,|L_N|\}&6&151{,}200\\
\text{signs}&256&524{,}288\\
\{L_1,\ldots,L_N\}&1536&79{,}272{,}345{,}600
\end{array}
\end{equation}

To find the $30$ conjugates of $A$, we first note that conjugation by a permutation is simply letting the permutation act on the labels of the cube.  The idea is to pick one such labeling for every conjugate, so we use $N(A)=A\rtimes \SL_3(\zt)$ to put the cube in a ``canonical'' position.

For instance, we can use translation by $A$ to put vertex $1$ at the lower front left (the ``origin'' of $\zt{}^3$).  We can then use a linear transformation to bring vertex $2$ to the lower front right, and vertex $3$ to the upper front left.  The result is the cube in Figure~\ref{fig:cubepart}.

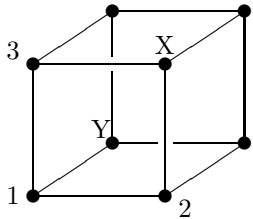
\begin{figure}[ht]
\begin{center}
\begin{picture}(100,100)(-10,-10)
\put(0,0){\circle*{5}}
\put(50,0){\circle*{5}}
\put(0,50){\circle*{5}}
\put(50,50){\circle*{5}}
\put(30,20){\circle*{5}}
\put(80,20){\circle*{5}}
\put(80,70){\circle*{5}}
\put(30,70){\circle*{5}}
\put(0,0){\line(1,0){50}}
\put(0,0){\line(0,1){50}}
\put(50,50){\line(-1,0){50}}
\put(50,50){\line(0,-1){50}}
\put(30,20){\line(1,0){17}}
\put(53,20){\line(1,0){27}}
\put(30,20){\line(0,1){27}}
\put(30,53){\line(0,1){17}}
\put(80,70){\line(-1,0){50}}
\put(80,70){\line(0,-1){50}}
\put(0,0){\line(3,2){30}}
\put(50,0){\line(3,2){30}}
\put(0,50){\line(3,2){30}}
\put(50,50){\line(3,2){30}}
\put(-10,52){3}
\put(46,54){X}
\put(-10,-3){1}
\put(55,-8){2}
\put(22,22){Y}
\end{picture}
\end{center}
\caption{Cube partly labeled.  We put $1$ at the ``origin'' and put $2$ and $3$ in standard basis position.  Positions X and Y are for reference in the discussion following.}
\label{fig:cubepart}
\end{figure}

At the next phase, we have two cases, depending on whether $5$ is in position X or not.  If $5$ is not in position X, we can use a linear transformation to bring it to position Y.  Then the linear transformation is determined, and there are $4!=24$ ways to distribute the other vertices.  If $5$ is in position X, then we can use a linear transformation to bring $4$ into position Y, and there are $3!=6$ ways to distribute the other vertices.

The resulting list of thirty cubes are shown in Figure~\ref{fig:thirty}.  For each, the set of translations modulo $2$ gives a conjugate of $A$.

\begin{figure}[ht]
{\scriptsize
\newlength{\templen}
\setlength{\templen}{\unitlength}
\setlength{\unitlength}{.007in}
\begin{center}
\begin{picture}(100,100)(-10,-10)
\put(0,0){\circle*{5}}
\put(50,0){\circle*{5}}
\put(0,50){\circle*{5}}
\put(50,50){\circle*{5}}
\put(30,20){\circle*{5}}
\put(80,20){\circle*{5}}
\put(80,70){\circle*{5}}
\put(30,70){\circle*{5}}
\put(0,0){\line(1,0){50}}
\put(0,0){\line(0,1){50}}
\put(50,50){\line(-1,0){50}}
\put(50,50){\line(0,-1){50}}
\put(30,20){\line(1,0){17}}
\put(53,20){\line(1,0){27}}
\put(30,20){\line(0,1){27}}
\put(30,53){\line(0,1){17}}
\put(80,70){\line(-1,0){50}}
\put(80,70){\line(0,-1){50}}
\put(0,0){\line(3,2){30}}
\put(50,0){\line(3,2){30}}
\put(0,50){\line(3,2){30}}
\put(50,50){\line(3,2){30}}
\put(-10,-3){1}
\put(55,-8){2}
\put(-10,52){3}
\put(46,54){4}
\put(22,22){5}
\put(82,22){6}
\put(22,73){7}
\put(82,73){8}
\end{picture}
\begin{picture}(100,100)(-10,-10)
\put(0,0){\circle*{5}}
\put(50,0){\circle*{5}}
\put(0,50){\circle*{5}}
\put(50,50){\circle*{5}}
\put(30,20){\circle*{5}}
\put(80,20){\circle*{5}}
\put(80,70){\circle*{5}}
\put(30,70){\circle*{5}}
\put(0,0){\line(1,0){50}}
\put(0,0){\line(0,1){50}}
\put(50,50){\line(-1,0){50}}
\put(50,50){\line(0,-1){50}}
\put(30,20){\line(1,0){17}}
\put(53,20){\line(1,0){27}}
\put(30,20){\line(0,1){27}}
\put(30,53){\line(0,1){17}}
\put(80,70){\line(-1,0){50}}
\put(80,70){\line(0,-1){50}}
\put(0,0){\line(3,2){30}}
\put(50,0){\line(3,2){30}}
\put(0,50){\line(3,2){30}}
\put(50,50){\line(3,2){30}}
\put(-10,-3){1}
\put(55,-8){2}
\put(-10,52){3}
\put(46,54){4}
\put(22,22){5}
\put(82,22){6}
\put(22,73){8}
\put(82,73){7}
\end{picture}
\begin{picture}(100,100)(-10,-10)
\put(0,0){\circle*{5}}
\put(50,0){\circle*{5}}
\put(0,50){\circle*{5}}
\put(50,50){\circle*{5}}
\put(30,20){\circle*{5}}
\put(80,20){\circle*{5}}
\put(80,70){\circle*{5}}
\put(30,70){\circle*{5}}
\put(0,0){\line(1,0){50}}
\put(0,0){\line(0,1){50}}
\put(50,50){\line(-1,0){50}}
\put(50,50){\line(0,-1){50}}
\put(30,20){\line(1,0){17}}
\put(53,20){\line(1,0){27}}
\put(30,20){\line(0,1){27}}
\put(30,53){\line(0,1){17}}
\put(80,70){\line(-1,0){50}}
\put(80,70){\line(0,-1){50}}
\put(0,0){\line(3,2){30}}
\put(50,0){\line(3,2){30}}
\put(0,50){\line(3,2){30}}
\put(50,50){\line(3,2){30}}
\put(-10,-3){1}
\put(55,-8){2}
\put(-10,52){3}
\put(46,54){4}
\put(22,22){5}
\put(82,22){7}
\put(22,73){6}
\put(82,73){8}
\end{picture}
\begin{picture}(100,100)(-10,-10)
\put(0,0){\circle*{5}}
\put(50,0){\circle*{5}}
\put(0,50){\circle*{5}}
\put(50,50){\circle*{5}}
\put(30,20){\circle*{5}}
\put(80,20){\circle*{5}}
\put(80,70){\circle*{5}}
\put(30,70){\circle*{5}}
\put(0,0){\line(1,0){50}}
\put(0,0){\line(0,1){50}}
\put(50,50){\line(-1,0){50}}
\put(50,50){\line(0,-1){50}}
\put(30,20){\line(1,0){17}}
\put(53,20){\line(1,0){27}}
\put(30,20){\line(0,1){27}}
\put(30,53){\line(0,1){17}}
\put(80,70){\line(-1,0){50}}
\put(80,70){\line(0,-1){50}}
\put(0,0){\line(3,2){30}}
\put(50,0){\line(3,2){30}}
\put(0,50){\line(3,2){30}}
\put(50,50){\line(3,2){30}}
\put(-10,-3){1}
\put(55,-8){2}
\put(-10,52){3}
\put(46,54){4}
\put(22,22){5}
\put(82,22){7}
\put(22,73){8}
\put(82,73){6}
\end{picture}
\begin{picture}(100,100)(-10,-10)
\put(0,0){\circle*{5}}
\put(50,0){\circle*{5}}
\put(0,50){\circle*{5}}
\put(50,50){\circle*{5}}
\put(30,20){\circle*{5}}
\put(80,20){\circle*{5}}
\put(80,70){\circle*{5}}
\put(30,70){\circle*{5}}
\put(0,0){\line(1,0){50}}
\put(0,0){\line(0,1){50}}
\put(50,50){\line(-1,0){50}}
\put(50,50){\line(0,-1){50}}
\put(30,20){\line(1,0){17}}
\put(53,20){\line(1,0){27}}
\put(30,20){\line(0,1){27}}
\put(30,53){\line(0,1){17}}
\put(80,70){\line(-1,0){50}}
\put(80,70){\line(0,-1){50}}
\put(0,0){\line(3,2){30}}
\put(50,0){\line(3,2){30}}
\put(0,50){\line(3,2){30}}
\put(50,50){\line(3,2){30}}
\put(-10,-3){1}
\put(55,-8){2}
\put(-10,52){3}
\put(46,54){4}
\put(22,22){5}
\put(82,22){8}
\put(22,73){6}
\put(82,73){7}
\end{picture}
\begin{picture}(100,100)(-10,-10)
\put(0,0){\circle*{5}}
\put(50,0){\circle*{5}}
\put(0,50){\circle*{5}}
\put(50,50){\circle*{5}}
\put(30,20){\circle*{5}}
\put(80,20){\circle*{5}}
\put(80,70){\circle*{5}}
\put(30,70){\circle*{5}}
\put(0,0){\line(1,0){50}}
\put(0,0){\line(0,1){50}}
\put(50,50){\line(-1,0){50}}
\put(50,50){\line(0,-1){50}}
\put(30,20){\line(1,0){17}}
\put(53,20){\line(1,0){27}}
\put(30,20){\line(0,1){27}}
\put(30,53){\line(0,1){17}}
\put(80,70){\line(-1,0){50}}
\put(80,70){\line(0,-1){50}}
\put(0,0){\line(3,2){30}}
\put(50,0){\line(3,2){30}}
\put(0,50){\line(3,2){30}}
\put(50,50){\line(3,2){30}}
\put(-10,-3){1}
\put(55,-8){2}
\put(-10,52){3}
\put(46,54){4}
\put(22,22){5}
\put(82,22){8}
\put(22,73){7}
\put(82,73){6}
\end{picture}
\begin{picture}(100,100)(-10,-10)
\put(0,0){\circle*{5}}
\put(50,0){\circle*{5}}
\put(0,50){\circle*{5}}
\put(50,50){\circle*{5}}
\put(30,20){\circle*{5}}
\put(80,20){\circle*{5}}
\put(80,70){\circle*{5}}
\put(30,70){\circle*{5}}
\put(0,0){\line(1,0){50}}
\put(0,0){\line(0,1){50}}
\put(50,50){\line(-1,0){50}}
\put(50,50){\line(0,-1){50}}
\put(30,20){\line(1,0){17}}
\put(53,20){\line(1,0){27}}
\put(30,20){\line(0,1){27}}
\put(30,53){\line(0,1){17}}
\put(80,70){\line(-1,0){50}}
\put(80,70){\line(0,-1){50}}
\put(0,0){\line(3,2){30}}
\put(50,0){\line(3,2){30}}
\put(0,50){\line(3,2){30}}
\put(50,50){\line(3,2){30}}
\put(-10,-3){1}
\put(55,-8){2}
\put(-10,52){3}
\put(46,54){6}
\put(22,22){5}
\put(82,22){4}
\put(22,73){7}
\put(82,73){8}
\end{picture}
\begin{picture}(100,100)(-10,-10)
\put(0,0){\circle*{5}}
\put(50,0){\circle*{5}}
\put(0,50){\circle*{5}}
\put(50,50){\circle*{5}}
\put(30,20){\circle*{5}}
\put(80,20){\circle*{5}}
\put(80,70){\circle*{5}}
\put(30,70){\circle*{5}}
\put(0,0){\line(1,0){50}}
\put(0,0){\line(0,1){50}}
\put(50,50){\line(-1,0){50}}
\put(50,50){\line(0,-1){50}}
\put(30,20){\line(1,0){17}}
\put(53,20){\line(1,0){27}}
\put(30,20){\line(0,1){27}}
\put(30,53){\line(0,1){17}}
\put(80,70){\line(-1,0){50}}
\put(80,70){\line(0,-1){50}}
\put(0,0){\line(3,2){30}}
\put(50,0){\line(3,2){30}}
\put(0,50){\line(3,2){30}}
\put(50,50){\line(3,2){30}}
\put(-10,-3){1}
\put(55,-8){2}
\put(-10,52){3}
\put(46,54){6}
\put(22,22){5}
\put(82,22){4}
\put(22,73){8}
\put(82,73){7}
\end{picture}
\begin{picture}(100,100)(-10,-10)
\put(0,0){\circle*{5}}
\put(50,0){\circle*{5}}
\put(0,50){\circle*{5}}
\put(50,50){\circle*{5}}
\put(30,20){\circle*{5}}
\put(80,20){\circle*{5}}
\put(80,70){\circle*{5}}
\put(30,70){\circle*{5}}
\put(0,0){\line(1,0){50}}
\put(0,0){\line(0,1){50}}
\put(50,50){\line(-1,0){50}}
\put(50,50){\line(0,-1){50}}
\put(30,20){\line(1,0){17}}
\put(53,20){\line(1,0){27}}
\put(30,20){\line(0,1){27}}
\put(30,53){\line(0,1){17}}
\put(80,70){\line(-1,0){50}}
\put(80,70){\line(0,-1){50}}
\put(0,0){\line(3,2){30}}
\put(50,0){\line(3,2){30}}
\put(0,50){\line(3,2){30}}
\put(50,50){\line(3,2){30}}
\put(-10,-3){1}
\put(55,-8){2}
\put(-10,52){3}
\put(46,54){6}
\put(22,22){5}
\put(82,22){7}
\put(22,73){4}
\put(82,73){8}
\end{picture}
\begin{picture}(100,100)(-10,-10)
\put(0,0){\circle*{5}}
\put(50,0){\circle*{5}}
\put(0,50){\circle*{5}}
\put(50,50){\circle*{5}}
\put(30,20){\circle*{5}}
\put(80,20){\circle*{5}}
\put(80,70){\circle*{5}}
\put(30,70){\circle*{5}}
\put(0,0){\line(1,0){50}}
\put(0,0){\line(0,1){50}}
\put(50,50){\line(-1,0){50}}
\put(50,50){\line(0,-1){50}}
\put(30,20){\line(1,0){17}}
\put(53,20){\line(1,0){27}}
\put(30,20){\line(0,1){27}}
\put(30,53){\line(0,1){17}}
\put(80,70){\line(-1,0){50}}
\put(80,70){\line(0,-1){50}}
\put(0,0){\line(3,2){30}}
\put(50,0){\line(3,2){30}}
\put(0,50){\line(3,2){30}}
\put(50,50){\line(3,2){30}}
\put(-10,-3){1}
\put(55,-8){2}
\put(-10,52){3}
\put(46,54){6}
\put(22,22){5}
\put(82,22){7}
\put(22,73){8}
\put(82,73){4}
\end{picture}
\begin{picture}(100,100)(-10,-10)
\put(0,0){\circle*{5}}
\put(50,0){\circle*{5}}
\put(0,50){\circle*{5}}
\put(50,50){\circle*{5}}
\put(30,20){\circle*{5}}
\put(80,20){\circle*{5}}
\put(80,70){\circle*{5}}
\put(30,70){\circle*{5}}
\put(0,0){\line(1,0){50}}
\put(0,0){\line(0,1){50}}
\put(50,50){\line(-1,0){50}}
\put(50,50){\line(0,-1){50}}
\put(30,20){\line(1,0){17}}
\put(53,20){\line(1,0){27}}
\put(30,20){\line(0,1){27}}
\put(30,53){\line(0,1){17}}
\put(80,70){\line(-1,0){50}}
\put(80,70){\line(0,-1){50}}
\put(0,0){\line(3,2){30}}
\put(50,0){\line(3,2){30}}
\put(0,50){\line(3,2){30}}
\put(50,50){\line(3,2){30}}
\put(-10,-3){1}
\put(55,-8){2}
\put(-10,52){3}
\put(46,54){6}
\put(22,22){5}
\put(82,22){8}
\put(22,73){4}
\put(82,73){7}
\end{picture}
\begin{picture}(100,100)(-10,-10)
\put(0,0){\circle*{5}}
\put(50,0){\circle*{5}}
\put(0,50){\circle*{5}}
\put(50,50){\circle*{5}}
\put(30,20){\circle*{5}}
\put(80,20){\circle*{5}}
\put(80,70){\circle*{5}}
\put(30,70){\circle*{5}}
\put(0,0){\line(1,0){50}}
\put(0,0){\line(0,1){50}}
\put(50,50){\line(-1,0){50}}
\put(50,50){\line(0,-1){50}}
\put(30,20){\line(1,0){17}}
\put(53,20){\line(1,0){27}}
\put(30,20){\line(0,1){27}}
\put(30,53){\line(0,1){17}}
\put(80,70){\line(-1,0){50}}
\put(80,70){\line(0,-1){50}}
\put(0,0){\line(3,2){30}}
\put(50,0){\line(3,2){30}}
\put(0,50){\line(3,2){30}}
\put(50,50){\line(3,2){30}}
\put(-10,-3){1}
\put(55,-8){2}
\put(-10,52){3}
\put(46,54){6}
\put(22,22){5}
\put(82,22){8}
\put(22,73){7}
\put(82,73){4}
\end{picture}
\begin{picture}(100,100)(-10,-10)
\put(0,0){\circle*{5}}
\put(50,0){\circle*{5}}
\put(0,50){\circle*{5}}
\put(50,50){\circle*{5}}
\put(30,20){\circle*{5}}
\put(80,20){\circle*{5}}
\put(80,70){\circle*{5}}
\put(30,70){\circle*{5}}
\put(0,0){\line(1,0){50}}
\put(0,0){\line(0,1){50}}
\put(50,50){\line(-1,0){50}}
\put(50,50){\line(0,-1){50}}
\put(30,20){\line(1,0){17}}
\put(53,20){\line(1,0){27}}
\put(30,20){\line(0,1){27}}
\put(30,53){\line(0,1){17}}
\put(80,70){\line(-1,0){50}}
\put(80,70){\line(0,-1){50}}
\put(0,0){\line(3,2){30}}
\put(50,0){\line(3,2){30}}
\put(0,50){\line(3,2){30}}
\put(50,50){\line(3,2){30}}
\put(-10,-3){1}
\put(55,-8){2}
\put(-10,52){3}
\put(46,54){7}
\put(22,22){5}
\put(82,22){4}
\put(22,73){6}
\put(82,73){8}
\end{picture}
\begin{picture}(100,100)(-10,-10)
\put(0,0){\circle*{5}}
\put(50,0){\circle*{5}}
\put(0,50){\circle*{5}}
\put(50,50){\circle*{5}}
\put(30,20){\circle*{5}}
\put(80,20){\circle*{5}}
\put(80,70){\circle*{5}}
\put(30,70){\circle*{5}}
\put(0,0){\line(1,0){50}}
\put(0,0){\line(0,1){50}}
\put(50,50){\line(-1,0){50}}
\put(50,50){\line(0,-1){50}}
\put(30,20){\line(1,0){17}}
\put(53,20){\line(1,0){27}}
\put(30,20){\line(0,1){27}}
\put(30,53){\line(0,1){17}}
\put(80,70){\line(-1,0){50}}
\put(80,70){\line(0,-1){50}}
\put(0,0){\line(3,2){30}}
\put(50,0){\line(3,2){30}}
\put(0,50){\line(3,2){30}}
\put(50,50){\line(3,2){30}}
\put(-10,-3){1}
\put(55,-8){2}
\put(-10,52){3}
\put(46,54){7}
\put(22,22){5}
\put(82,22){4}
\put(22,73){8}
\put(82,73){6}
\end{picture}
\begin{picture}(100,100)(-10,-10)
\put(0,0){\circle*{5}}
\put(50,0){\circle*{5}}
\put(0,50){\circle*{5}}
\put(50,50){\circle*{5}}
\put(30,20){\circle*{5}}
\put(80,20){\circle*{5}}
\put(80,70){\circle*{5}}
\put(30,70){\circle*{5}}
\put(0,0){\line(1,0){50}}
\put(0,0){\line(0,1){50}}
\put(50,50){\line(-1,0){50}}
\put(50,50){\line(0,-1){50}}
\put(30,20){\line(1,0){17}}
\put(53,20){\line(1,0){27}}
\put(30,20){\line(0,1){27}}
\put(30,53){\line(0,1){17}}
\put(80,70){\line(-1,0){50}}
\put(80,70){\line(0,-1){50}}
\put(0,0){\line(3,2){30}}
\put(50,0){\line(3,2){30}}
\put(0,50){\line(3,2){30}}
\put(50,50){\line(3,2){30}}
\put(-10,-3){1}
\put(55,-8){2}
\put(-10,52){3}
\put(46,54){7}
\put(22,22){5}
\put(82,22){6}
\put(22,73){4}
\put(82,73){8}
\end{picture}
\begin{picture}(100,100)(-10,-10)
\put(0,0){\circle*{5}}
\put(50,0){\circle*{5}}
\put(0,50){\circle*{5}}
\put(50,50){\circle*{5}}
\put(30,20){\circle*{5}}
\put(80,20){\circle*{5}}
\put(80,70){\circle*{5}}
\put(30,70){\circle*{5}}
\put(0,0){\line(1,0){50}}
\put(0,0){\line(0,1){50}}
\put(50,50){\line(-1,0){50}}
\put(50,50){\line(0,-1){50}}
\put(30,20){\line(1,0){17}}
\put(53,20){\line(1,0){27}}
\put(30,20){\line(0,1){27}}
\put(30,53){\line(0,1){17}}
\put(80,70){\line(-1,0){50}}
\put(80,70){\line(0,-1){50}}
\put(0,0){\line(3,2){30}}
\put(50,0){\line(3,2){30}}
\put(0,50){\line(3,2){30}}
\put(50,50){\line(3,2){30}}
\put(-10,-3){1}
\put(55,-8){2}
\put(-10,52){3}
\put(46,54){7}
\put(22,22){5}
\put(82,22){6}
\put(22,73){8}
\put(82,73){4}
\end{picture}
\begin{picture}(100,100)(-10,-10)
\put(0,0){\circle*{5}}
\put(50,0){\circle*{5}}
\put(0,50){\circle*{5}}
\put(50,50){\circle*{5}}
\put(30,20){\circle*{5}}
\put(80,20){\circle*{5}}
\put(80,70){\circle*{5}}
\put(30,70){\circle*{5}}
\put(0,0){\line(1,0){50}}
\put(0,0){\line(0,1){50}}
\put(50,50){\line(-1,0){50}}
\put(50,50){\line(0,-1){50}}
\put(30,20){\line(1,0){17}}
\put(53,20){\line(1,0){27}}
\put(30,20){\line(0,1){27}}
\put(30,53){\line(0,1){17}}
\put(80,70){\line(-1,0){50}}
\put(80,70){\line(0,-1){50}}
\put(0,0){\line(3,2){30}}
\put(50,0){\line(3,2){30}}
\put(0,50){\line(3,2){30}}
\put(50,50){\line(3,2){30}}
\put(-10,-3){1}
\put(55,-8){2}
\put(-10,52){3}
\put(46,54){7}
\put(22,22){5}
\put(82,22){8}
\put(22,73){4}
\put(82,73){6}
\end{picture}
\begin{picture}(100,100)(-10,-10)
\put(0,0){\circle*{5}}
\put(50,0){\circle*{5}}
\put(0,50){\circle*{5}}
\put(50,50){\circle*{5}}
\put(30,20){\circle*{5}}
\put(80,20){\circle*{5}}
\put(80,70){\circle*{5}}
\put(30,70){\circle*{5}}
\put(0,0){\line(1,0){50}}
\put(0,0){\line(0,1){50}}
\put(50,50){\line(-1,0){50}}
\put(50,50){\line(0,-1){50}}
\put(30,20){\line(1,0){17}}
\put(53,20){\line(1,0){27}}
\put(30,20){\line(0,1){27}}
\put(30,53){\line(0,1){17}}
\put(80,70){\line(-1,0){50}}
\put(80,70){\line(0,-1){50}}
\put(0,0){\line(3,2){30}}
\put(50,0){\line(3,2){30}}
\put(0,50){\line(3,2){30}}
\put(50,50){\line(3,2){30}}
\put(-10,-3){1}
\put(55,-8){2}
\put(-10,52){3}
\put(46,54){7}
\put(22,22){5}
\put(82,22){8}
\put(22,73){6}
\put(82,73){4}
\end{picture}
\begin{picture}(100,100)(-10,-10)
\put(0,0){\circle*{5}}
\put(50,0){\circle*{5}}
\put(0,50){\circle*{5}}
\put(50,50){\circle*{5}}
\put(30,20){\circle*{5}}
\put(80,20){\circle*{5}}
\put(80,70){\circle*{5}}
\put(30,70){\circle*{5}}
\put(0,0){\line(1,0){50}}
\put(0,0){\line(0,1){50}}
\put(50,50){\line(-1,0){50}}
\put(50,50){\line(0,-1){50}}
\put(30,20){\line(1,0){17}}
\put(53,20){\line(1,0){27}}
\put(30,20){\line(0,1){27}}
\put(30,53){\line(0,1){17}}
\put(80,70){\line(-1,0){50}}
\put(80,70){\line(0,-1){50}}
\put(0,0){\line(3,2){30}}
\put(50,0){\line(3,2){30}}
\put(0,50){\line(3,2){30}}
\put(50,50){\line(3,2){30}}
\put(-10,-3){1}
\put(55,-8){2}
\put(-10,52){3}
\put(46,54){8}
\put(22,22){5}
\put(82,22){4}
\put(22,73){6}
\put(82,73){7}
\end{picture}
\begin{picture}(100,100)(-10,-10)
\put(0,0){\circle*{5}}
\put(50,0){\circle*{5}}
\put(0,50){\circle*{5}}
\put(50,50){\circle*{5}}
\put(30,20){\circle*{5}}
\put(80,20){\circle*{5}}
\put(80,70){\circle*{5}}
\put(30,70){\circle*{5}}
\put(0,0){\line(1,0){50}}
\put(0,0){\line(0,1){50}}
\put(50,50){\line(-1,0){50}}
\put(50,50){\line(0,-1){50}}
\put(30,20){\line(1,0){17}}
\put(53,20){\line(1,0){27}}
\put(30,20){\line(0,1){27}}
\put(30,53){\line(0,1){17}}
\put(80,70){\line(-1,0){50}}
\put(80,70){\line(0,-1){50}}
\put(0,0){\line(3,2){30}}
\put(50,0){\line(3,2){30}}
\put(0,50){\line(3,2){30}}
\put(50,50){\line(3,2){30}}
\put(-10,-3){1}
\put(55,-8){2}
\put(-10,52){3}
\put(46,54){8}
\put(22,22){5}
\put(82,22){4}
\put(22,73){7}
\put(82,73){6}
\end{picture}
\begin{picture}(100,100)(-10,-10)
\put(0,0){\circle*{5}}
\put(50,0){\circle*{5}}
\put(0,50){\circle*{5}}
\put(50,50){\circle*{5}}
\put(30,20){\circle*{5}}
\put(80,20){\circle*{5}}
\put(80,70){\circle*{5}}
\put(30,70){\circle*{5}}
\put(0,0){\line(1,0){50}}
\put(0,0){\line(0,1){50}}
\put(50,50){\line(-1,0){50}}
\put(50,50){\line(0,-1){50}}
\put(30,20){\line(1,0){17}}
\put(53,20){\line(1,0){27}}
\put(30,20){\line(0,1){27}}
\put(30,53){\line(0,1){17}}
\put(80,70){\line(-1,0){50}}
\put(80,70){\line(0,-1){50}}
\put(0,0){\line(3,2){30}}
\put(50,0){\line(3,2){30}}
\put(0,50){\line(3,2){30}}
\put(50,50){\line(3,2){30}}
\put(-10,-3){1}
\put(55,-8){2}
\put(-10,52){3}
\put(46,54){8}
\put(22,22){5}
\put(82,22){6}
\put(22,73){4}
\put(82,73){7}
\end{picture}
\begin{picture}(100,100)(-10,-10)
\put(0,0){\circle*{5}}
\put(50,0){\circle*{5}}
\put(0,50){\circle*{5}}
\put(50,50){\circle*{5}}
\put(30,20){\circle*{5}}
\put(80,20){\circle*{5}}
\put(80,70){\circle*{5}}
\put(30,70){\circle*{5}}
\put(0,0){\line(1,0){50}}
\put(0,0){\line(0,1){50}}
\put(50,50){\line(-1,0){50}}
\put(50,50){\line(0,-1){50}}
\put(30,20){\line(1,0){17}}
\put(53,20){\line(1,0){27}}
\put(30,20){\line(0,1){27}}
\put(30,53){\line(0,1){17}}
\put(80,70){\line(-1,0){50}}
\put(80,70){\line(0,-1){50}}
\put(0,0){\line(3,2){30}}
\put(50,0){\line(3,2){30}}
\put(0,50){\line(3,2){30}}
\put(50,50){\line(3,2){30}}
\put(-10,-3){1}
\put(55,-8){2}
\put(-10,52){3}
\put(46,54){8}
\put(22,22){5}
\put(82,22){6}
\put(22,73){7}
\put(82,73){4}
\end{picture}
\begin{picture}(100,100)(-10,-10)
\put(0,0){\circle*{5}}
\put(50,0){\circle*{5}}
\put(0,50){\circle*{5}}
\put(50,50){\circle*{5}}
\put(30,20){\circle*{5}}
\put(80,20){\circle*{5}}
\put(80,70){\circle*{5}}
\put(30,70){\circle*{5}}
\put(0,0){\line(1,0){50}}
\put(0,0){\line(0,1){50}}
\put(50,50){\line(-1,0){50}}
\put(50,50){\line(0,-1){50}}
\put(30,20){\line(1,0){17}}
\put(53,20){\line(1,0){27}}
\put(30,20){\line(0,1){27}}
\put(30,53){\line(0,1){17}}
\put(80,70){\line(-1,0){50}}
\put(80,70){\line(0,-1){50}}
\put(0,0){\line(3,2){30}}
\put(50,0){\line(3,2){30}}
\put(0,50){\line(3,2){30}}
\put(50,50){\line(3,2){30}}
\put(-10,-3){1}
\put(55,-8){2}
\put(-10,52){3}
\put(46,54){8}
\put(22,22){5}
\put(82,22){7}
\put(22,73){4}
\put(82,73){6}
\end{picture}
\begin{picture}(100,100)(-10,-10)
\put(0,0){\circle*{5}}
\put(50,0){\circle*{5}}
\put(0,50){\circle*{5}}
\put(50,50){\circle*{5}}
\put(30,20){\circle*{5}}
\put(80,20){\circle*{5}}
\put(80,70){\circle*{5}}
\put(30,70){\circle*{5}}
\put(0,0){\line(1,0){50}}
\put(0,0){\line(0,1){50}}
\put(50,50){\line(-1,0){50}}
\put(50,50){\line(0,-1){50}}
\put(30,20){\line(1,0){17}}
\put(53,20){\line(1,0){27}}
\put(30,20){\line(0,1){27}}
\put(30,53){\line(0,1){17}}
\put(80,70){\line(-1,0){50}}
\put(80,70){\line(0,-1){50}}
\put(0,0){\line(3,2){30}}
\put(50,0){\line(3,2){30}}
\put(0,50){\line(3,2){30}}
\put(50,50){\line(3,2){30}}
\put(-10,-3){1}
\put(55,-8){2}
\put(-10,52){3}
\put(46,54){8}
\put(22,22){5}
\put(82,22){7}
\put(22,73){6}
\put(82,73){4}
\end{picture}
\begin{picture}(100,100)(-10,-10)
\put(0,0){\circle*{5}}
\put(50,0){\circle*{5}}
\put(0,50){\circle*{5}}
\put(50,50){\circle*{5}}
\put(30,20){\circle*{5}}
\put(80,20){\circle*{5}}
\put(80,70){\circle*{5}}
\put(30,70){\circle*{5}}
\put(0,0){\line(1,0){50}}
\put(0,0){\line(0,1){50}}
\put(50,50){\line(-1,0){50}}
\put(50,50){\line(0,-1){50}}
\put(30,20){\line(1,0){17}}
\put(53,20){\line(1,0){27}}
\put(30,20){\line(0,1){27}}
\put(30,53){\line(0,1){17}}
\put(80,70){\line(-1,0){50}}
\put(80,70){\line(0,-1){50}}
\put(0,0){\line(3,2){30}}
\put(50,0){\line(3,2){30}}
\put(0,50){\line(3,2){30}}
\put(50,50){\line(3,2){30}}
\put(-10,-3){1}
\put(55,-8){2}
\put(-10,52){3}
\put(46,54){5}
\put(22,22){4}
\put(82,22){6}
\put(22,73){7}
\put(82,73){8}
\end{picture}
\begin{picture}(100,100)(-10,-10)
\put(0,0){\circle*{5}}
\put(50,0){\circle*{5}}
\put(0,50){\circle*{5}}
\put(50,50){\circle*{5}}
\put(30,20){\circle*{5}}
\put(80,20){\circle*{5}}
\put(80,70){\circle*{5}}
\put(30,70){\circle*{5}}
\put(0,0){\line(1,0){50}}
\put(0,0){\line(0,1){50}}
\put(50,50){\line(-1,0){50}}
\put(50,50){\line(0,-1){50}}
\put(30,20){\line(1,0){17}}
\put(53,20){\line(1,0){27}}
\put(30,20){\line(0,1){27}}
\put(30,53){\line(0,1){17}}
\put(80,70){\line(-1,0){50}}
\put(80,70){\line(0,-1){50}}
\put(0,0){\line(3,2){30}}
\put(50,0){\line(3,2){30}}
\put(0,50){\line(3,2){30}}
\put(50,50){\line(3,2){30}}
\put(-10,-3){1}
\put(55,-8){2}
\put(-10,52){3}
\put(46,54){5}
\put(22,22){4}
\put(82,22){6}
\put(22,73){8}
\put(82,73){7}
\end{picture}
\begin{picture}(100,100)(-10,-10)
\put(0,0){\circle*{5}}
\put(50,0){\circle*{5}}
\put(0,50){\circle*{5}}
\put(50,50){\circle*{5}}
\put(30,20){\circle*{5}}
\put(80,20){\circle*{5}}
\put(80,70){\circle*{5}}
\put(30,70){\circle*{5}}
\put(0,0){\line(1,0){50}}
\put(0,0){\line(0,1){50}}
\put(50,50){\line(-1,0){50}}
\put(50,50){\line(0,-1){50}}
\put(30,20){\line(1,0){17}}
\put(53,20){\line(1,0){27}}
\put(30,20){\line(0,1){27}}
\put(30,53){\line(0,1){17}}
\put(80,70){\line(-1,0){50}}
\put(80,70){\line(0,-1){50}}
\put(0,0){\line(3,2){30}}
\put(50,0){\line(3,2){30}}
\put(0,50){\line(3,2){30}}
\put(50,50){\line(3,2){30}}
\put(-10,-3){1}
\put(55,-8){2}
\put(-10,52){3}
\put(46,54){5}
\put(22,22){4}
\put(82,22){7}
\put(22,73){6}
\put(82,73){8}
\end{picture}
\begin{picture}(100,100)(-10,-10)
\put(0,0){\circle*{5}}
\put(50,0){\circle*{5}}
\put(0,50){\circle*{5}}
\put(50,50){\circle*{5}}
\put(30,20){\circle*{5}}
\put(80,20){\circle*{5}}
\put(80,70){\circle*{5}}
\put(30,70){\circle*{5}}
\put(0,0){\line(1,0){50}}
\put(0,0){\line(0,1){50}}
\put(50,50){\line(-1,0){50}}
\put(50,50){\line(0,-1){50}}
\put(30,20){\line(1,0){17}}
\put(53,20){\line(1,0){27}}
\put(30,20){\line(0,1){27}}
\put(30,53){\line(0,1){17}}
\put(80,70){\line(-1,0){50}}
\put(80,70){\line(0,-1){50}}
\put(0,0){\line(3,2){30}}
\put(50,0){\line(3,2){30}}
\put(0,50){\line(3,2){30}}
\put(50,50){\line(3,2){30}}
\put(-10,-3){1}
\put(55,-8){2}
\put(-10,52){3}
\put(46,54){5}
\put(22,22){4}
\put(82,22){7}
\put(22,73){8}
\put(82,73){6}
\end{picture}
\begin{picture}(100,100)(-10,-10)
\put(0,0){\circle*{5}}
\put(50,0){\circle*{5}}
\put(0,50){\circle*{5}}
\put(50,50){\circle*{5}}
\put(30,20){\circle*{5}}
\put(80,20){\circle*{5}}
\put(80,70){\circle*{5}}
\put(30,70){\circle*{5}}
\put(0,0){\line(1,0){50}}
\put(0,0){\line(0,1){50}}
\put(50,50){\line(-1,0){50}}
\put(50,50){\line(0,-1){50}}
\put(30,20){\line(1,0){17}}
\put(53,20){\line(1,0){27}}
\put(30,20){\line(0,1){27}}
\put(30,53){\line(0,1){17}}
\put(80,70){\line(-1,0){50}}
\put(80,70){\line(0,-1){50}}
\put(0,0){\line(3,2){30}}
\put(50,0){\line(3,2){30}}
\put(0,50){\line(3,2){30}}
\put(50,50){\line(3,2){30}}
\put(-10,-3){1}
\put(55,-8){2}
\put(-10,52){3}
\put(46,54){5}
\put(22,22){4}
\put(82,22){8}
\put(22,73){6}
\put(82,73){7}
\end{picture}
\begin{picture}(100,100)(-10,-10)
\put(0,0){\circle*{5}}
\put(50,0){\circle*{5}}
\put(0,50){\circle*{5}}
\put(50,50){\circle*{5}}
\put(30,20){\circle*{5}}
\put(80,20){\circle*{5}}
\put(80,70){\circle*{5}}
\put(30,70){\circle*{5}}
\put(0,0){\line(1,0){50}}
\put(0,0){\line(0,1){50}}
\put(50,50){\line(-1,0){50}}
\put(50,50){\line(0,-1){50}}
\put(30,20){\line(1,0){17}}
\put(53,20){\line(1,0){27}}
\put(30,20){\line(0,1){27}}
\put(30,53){\line(0,1){17}}
\put(80,70){\line(-1,0){50}}
\put(80,70){\line(0,-1){50}}
\put(0,0){\line(3,2){30}}
\put(50,0){\line(3,2){30}}
\put(0,50){\line(3,2){30}}
\put(50,50){\line(3,2){30}}
\put(-10,-3){1}
\put(55,-8){2}
\put(-10,52){3}
\put(46,54){5}
\put(22,22){4}
\put(82,22){8}
\put(22,73){7}
\put(82,73){6}
\end{picture}
\end{center}
\setlength{\unitlength}{\templen}
}
\caption{The $30$ ways to label the cube, up to affine transformations.  These correspond to the $30$ conjugates of $A$ in $\sym_8$.}
\label{fig:thirty}
\end{figure}
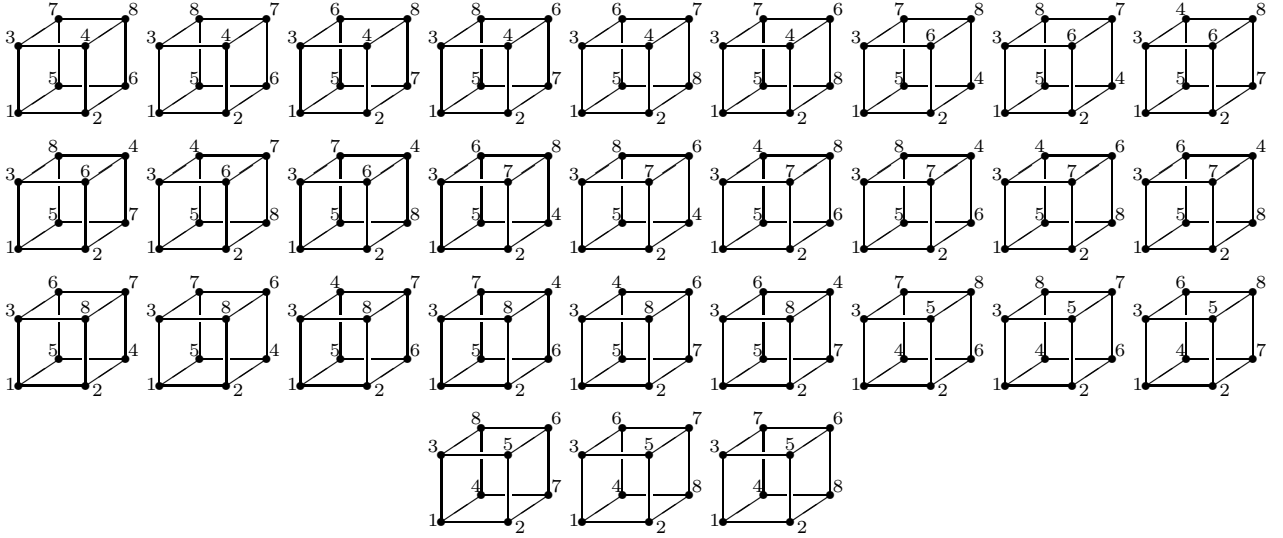

\subsection{Duality}
Again, as in $N=4$, we seek to identify
\[ \{c_1,\ldots,c_8\} \]
that form a set of permutation matrices so that
\[ \{c_1{}^{-1},\ldots,c_8{}^{-1}\}=\{c_1,\ldots,c_8\}. \]
Since the former is described as
\[ \pi B \]
for $B$ a conjugate of $A$, we need
\begin{equation}
(\pi B)^{-1}=\pi B.
\label{inversetwo}
\end{equation}
Now $(\pi B)^{-1}=B\pi^{-1}$, a right coset.  So the condition that $\pi B$ be self-dual implies that the right coset $B\pi^{-1}$ is equal to the left coset $\pi B$.  Since $B$ is not a normal subgroup, $B\pi^{-1}$ is not always a left coset at all, but if it is one, it must be equal to $\pi^{-1}B$, since they both contain $\pi^{-1}$ and since left cosets of $B$ partition $\sym_8$.  Then $\pi^{-1}$, and therefore $\pi$, is in $N(B)$, the normalizer of $B$, and $\pi B$ is an element of order $1$ or $2$ in $N(B)/B$.

Since there is an automorphism of $\sym_8$ that sends $A$ to $B$ (namely, conjugation), we can study the normalizer of $B$ by studying the normalizer of $A$.  For the rest of this discussion, we use $A$ instead of $B$.

The quotient $N(A)/A$ is $\SL_3(\zt)$.  The identity is the only element of order $1$, and there are $21$ elements of order $2$, all of them conjugate to each other in $\SL_3(\zt)$.\cite{ATLAS,PSL3}

For instance, one of them is
\[ \left[\begin{array}{ccc}
0&1&0\\
1&0&0\\
0&0&1
\end{array}\right] \]
As a permutation this turns into
\[ \pi=(23)(67). \]
The corresponding left coset is
\begin{equation}
\begin{aligned}
\pi a_1&=(2\,3)(6\,7),\\
\pi a_2&=(1\,3\,4\,2)(5\,7\,8\,6),\\
\pi a_3&=(1\,2\,4\,3)(5\,6\,8\,7),\\
\pi a_4&=(1\,5)(2\,7)(3\,6)(4\,8),\\
\pi a_5&=(1\,8)(2\,6)(3\,7)(4\,5),\\
\pi a_6&=(1\,6\,4\,7)(2\,8\,3\,5),\\
\pi a_7&=(1\,7\,4\,6)(2\,5\,3\,8),\\
\pi a_8&=(1\,4)(5\,8).
\end{aligned}
 \label{e:pA}
\end{equation}
We see that $\pi a_1$, $\pi a_4$, $\pi a_5$, and $\pi a_8$ all are self-dual, while $\pi a_2$ and $\pi a_3$ are duals of each other and $\pi a_7$ and $\pi a_8$ are duals of each other.

There are $22\cdot 30=660$ such self-dual $\pi B$ classes.  The $22$ cosets of $A$ that are self-dual are shown in Figure~\ref{fig:twentytwo}.  A similar list could be done for each of the other $29$ conjugates of $A$. The remaining unsigned permutations form $75{,}270$ dual pairs.

Again, each element in the 8-tiple of $L_I$-matrices represented by the group $A$ itself~(\ref{e:A}) is self-dual. In turn, the 8-tiple $(23)(67)A$~(\ref{e:pA}) also has order-4 elements, which are not self-dual; however, the dual of every element in $(23)(67)A$ is again an element in the same 8-tuple, so that this 8-tuple is self-dual. In particular, in this example we have that duality swaps $L_2\leftrightarrow L_3$ and $L_6\leftrightarrow L_7$.
 As with the $N=4$ supermultiplets $\VM,\VM_1$ and $\VM_2$, this means that the $N=8$ supersymmetry action in the $(23)(67)A$ 8-tuple differs from that in its dual, so that the $(23)(67)A$ supermultiplet and its dual are ``usefully distinct.''
 An exhaustive characterization of the duality action\footnote{For example, one can further partition the self-dual 8-tuples according to how many, or (in {\em\/ordered\/} 8-tuples) even precisely which $L_I$'s are permuted; this determines which of the corresponding supermultiplets are ``usefully distinct.'' Similar refinement in this combinatorial partitioning can also be done within the dual pairs.} within the complete listing of 8-tuplets of $L_I$-matrices realizing the $\GR(8,8)$ algebra is clearly a finite, but formidable task and beyond the scope of this article.

\begin{figure}[htp]
\begin{tabular}{c}
$\left[\begin{array}{ccc}
1&0&0\\
0&1&0\\
0&0&1
\end{array}\right]$\\
$(\,)$
\end{tabular}
\begin{tabular}{c}
$\left[\begin{array}{ccc}
0&1&0\\
1&0&0\\
0&0&1
\end{array}\right]$\\
$(23)(67)$\\
\end{tabular}
\begin{tabular}{c}
$\left[\begin{array}{ccc}
0&0&1\\
0&1&0\\
1&0&0
\end{array}\right]$\\
$(25)(47)$
\end{tabular}
\begin{tabular}{c}
$\left[\begin{array}{ccc}
1&0&0\\
0&0&1\\
0&1&0
\end{array}\right]$\\
$(35)(46)$
\end{tabular}
\sk

\begin{tabular}{c}
$\left[\begin{array}{ccc}
1&0&0\\
1&0&1\\
1&1&0
\end{array}\right]$\\
$(28)(35)$
\end{tabular}
\begin{tabular}{c}
$\left[\begin{array}{ccc}
1&0&0\\
1&1&0\\
1&0&1
\end{array}\right]$\\
$(28)(46)$
\end{tabular}
\begin{tabular}{c}
$\left[\begin{array}{ccc}
0&1&1\\
0&1&0\\
1&1&0
\end{array}\right]$\\
$(25)(38)$
\end{tabular}
\begin{tabular}{c}
$\left[\begin{array}{ccc}
1&1&0\\
0&1&0\\
0&1&1
\end{array}\right]$\\
$(38)(47)$
\end{tabular}
\begin{tabular}{c}
$\left[\begin{array}{ccc}
1&0&1\\
0&1&1\\
0&0&1
\end{array}\right]$\\
$(58)(67)$
\end{tabular}
\begin{tabular}{c}
$\left[\begin{array}{ccc}
0&1&1\\
1&0&1\\
0&0&1
\end{array}\right]$\\
$(23)(58)$
\end{tabular}
\sk

\begin{tabular}{c}
$\left[\begin{array}{ccc}
1&1&1\\
0&1&0\\
0&0&1
\end{array}\right]$\\
$(34)(56)$
\end{tabular}
\begin{tabular}{c}
$\left[\begin{array}{ccc}
1&1&1\\
0&0&1\\
0&1&0
\end{array}\right]$\\
$(36)(45)$
\end{tabular}
\begin{tabular}{c}
$\left[\begin{array}{ccc}
1&0&0\\
1&1&1\\
0&0&1
\end{array}\right]$\\
$(24)(57)$
\end{tabular}
\begin{tabular}{c}
$\left[\begin{array}{ccc}
0&0&1\\
1&1&1\\
1&0&0
\end{array}\right]$\\
$(27)(45)$
\end{tabular}
\begin{tabular}{c}
$\left[\begin{array}{ccc}
1&0&0\\
0&1&0\\
1&1&1
\end{array}\right]$\\
$(26)(37)$
\end{tabular}
\begin{tabular}{c}
$\left[\begin{array}{ccc}
0&1&0\\
1&0&0\\
1&1&1
\end{array}\right]$\\
$(27)(36)$
\end{tabular}

\sk

\begin{tabular}{c}
$\left[\begin{array}{ccc}
1&1&0\\
0&1&0\\
0&0&1
\end{array}\right]$\\
$(34)(78)$
\end{tabular}
\begin{tabular}{c}
$\left[\begin{array}{ccc}
1&0&1\\
0&1&0\\
0&0&1
\end{array}\right]$\\
$(56)(78)$
\end{tabular}
\begin{tabular}{c}
$\left[\begin{array}{ccc}
1&0&0\\
0&1&1\\
0&0&1
\end{array}\right]$\\
$(57)(68)$
\end{tabular}
\begin{tabular}{c}
$\left[\begin{array}{ccc}
1&0&0\\
1&1&0\\
0&0&1
\end{array}\right]$\\
$(24)(68)$
\end{tabular}
\begin{tabular}{c}
$\left[\begin{array}{ccc}
1&0&0\\
0&1&0\\
1&0&1
\end{array}\right]$\\
$(26)(48)$
\end{tabular}
\begin{tabular}{c}
$\left[\begin{array}{ccc}
1&0&0\\
0&1&0\\
0&1&1
\end{array}\right]$\\
$(37)(48)$
\end{tabular}

\caption{The $22$ elements of $N(A)/A$ of order $1$ or $2$, as elements of $\SL_3(\zt)$ and as permutations.  These correspond to self-dual left-cosets of $A$.  For the other conjugates of $A$, there is a similar list.}
\label{fig:twentytwo}
\end{figure}

\section{Other values of $N$}
For $\GR(4,4)$ and $\GR(8,8)$, we had $d=d_{\min}=N$.  This is also possible for $N=1$ (where $d_{\min}=1$) and $N=2$ (where $d_{\min}=2$), and in these cases, the same analysis works, although the result is rather trivial since there is only one set of permutation matrices.

For other values of $N$, we must have $d>N$, and this approach does not so neatly characterize the sets of permutation matrices $\{|L_1|,\ldots,|L_N|\}$.

For instance, if $N=3$, then $d_{\min}=4$, the set $\{|L_1|,|L_2|,|L_3|\}$ is never forms a group, or a coset of a group.

A more instructive example is $N=16$.  Here, $d_{\min}=128$, and furthermore there are two doubly even codes that give this $d_{\min}$: $e_8\oplus e_8$ and $d_{16}{}^+$.\footnote{See pp.~366--373 of Ref.~\citen{rCHVP} in the classification of self-dual codes, noting that doubly even codes are self-dual codes of Type II.  By Construction A, the lattices that correspond to these are the $E_8\times E_8$ and $SO(32)$ lattices well known to string theorists.}  The fact that there are two codes is not much of an obstacle: we analyze each case separately.  For instance, suppose we start with $e_8\oplus e_8$, which has, as generators, rows of the matrix:
\[ \left[\begin{array}{cccccccccccccccc}
0&1&1&1&1&0&0&0&0&0&0&0&0&0&0&0\\
1&0&1&1&0&1&0&0&0&0&0&0&0&0&0&0\\
1&1&0&1&0&0&1&0&0&0&0&0&0&0&0&0\\
1&1&1&0&0&0&0&1&0&0&0&0&0&0&0&0\\
0&0&0&0&0&0&0&0&0&1&1&1&1&0&0&0\\
0&0&0&0&0&0&0&0&1&0&1&1&0&1&0&0\\
0&0&0&0&0&0&0&0&1&1&0&1&0&0&1&0\\
0&0&0&0&0&0&0&0&1&1&1&0&0&0&0&1
\end{array}\right]
\]

There are indeed Valise Adinkras for both of these codes.  For example, one such choice involves permutations such as these:
\begin{equation}
\begin{array}{rlrl}
s_1&=(\,),\\
s_2&\MC3l{=(1\,2)(3\,4)(5\,6)(7\,8)\cdots(125\,126)(127\,128), }\\
s_3&\MC3l{=(1\,3)(2\,4)(5\,7)(6\,8)\cdots(125\,127)(126\,128),}\\
s_4&\MC3l{=(1\,5)(2\,6)(3\,7)(4\,8)\cdots(123\,127)(124\,128),}\\
s_9&\MC3l{=(1\,9)(2\,10)(3\,11)(4\,12)\cdots(119\,127)(120\,128),}\\
s_{10}&\MC3l{=(1\,17)(2\,18)(3\,19)(4\,20)\cdots(111\,127)(112\,128),}\\
s_{11}&\MC3l{=(1\,33)(2\,34)(3\,35)(4\,36)\cdots(95\,127)(96\,128),}\\
s_{12}&\MC3l{=(1\,65)(2\,66)(3\,67)(4\,68)\cdots(63\,127)(64\,128),}\\
s_5&=s_2s_3s_4,&\qquad\qquad s_{13}&=s_{10}s_{11}s_{12},\\
s_6&=s_3s_4,&\qquad\qquad   s_{14}&=s_9s_{11}s_{12},\\
s_7&=s_2s_4,&\qquad\qquad   s_{15}&=s_9s_{10}s_{12},\\
s_8&=s_2s_3,&\qquad\qquad   s_{16}&=s_9s_{10}s_{11}.\\
\end{array}
\end{equation}

Even though this set contains the identity, $s_1=(\,)$, $\{s_1,\ldots,s_{16}\}$ is not a group: for instance, $s_2s_9$ is not $s_i$ for any $i$.  Thus, the approach in this paper is not sufficient to understand these more general cases.  Recent work by the fourth author of this paper suggests another view of these groups that allow generalization to other cases, and will be addressed in a subsequent paper.



There may be other approaches to doing the counting:  Yan X. Zhang has a promising work in progress that takes a very different approach to this kind of count.\cite{YanTalk}

\section{One, two, four, eight}
The fact that this approach works for $\GR(1,1)$, $\GR(2,2)$, $\GR(4,4)$, and $\GR(8,8)$ is reminiscent of the existence of real division algebras in precisely the dimensions $1$, $2$, $4$, and $8$: the reals, the complex numbers, the quaternions, and the octonions.\cite{hur1,hur2}  Recall that for the division algebras, we lose something each time we go up in dimension: first order, then commutativity, then associativity.  Carrying this procedure further gives algebras that do not have inverses.\cite{numbers,booknumbers}

Likewise, comparing the set $\{|L_I|\}$ for various $\GR(N,N)$: for $\GR(1,1)$ this set is trivial.  For $\GR(2,2)$ it becomes non-trivial but is the entire symmetric group $\sym_2$.  For $\GR(4,4)$ it is no longer all of $\sym_4$, but is a coset of a normal subgroup.  For $\GR(8,8)$, it is a coset of a subgroup of $\sym_8$ but the subgroup is not a normal subgroup.  Carrying this procedure for higher $N$, the set $\{|L_I|\}$ is not a coset of a subgroup at all.

The classification of division algebras has been tied to supersymmetry for some time, with the work of Kugo and Townsend\cite{KugoTownsend} and  work of Baez and Huerta,\cite{BaezHuerta1} to name but a few for starters.  Further work is needed to elucidate the relationship between those ideas and the present work.

\section{The meaning of labeling vertices}
\label{sec:metaphysics}
In Section~\ref{sec:adinkra}, we mentioned the distinction between classifying sets of $L_I$ matrices and classifying Adinkras.  This was important in Section~\ref{sec:n4}, where we classified $\GR(4,4)$ algebras.  The main difference is that with Adinkras, if we rearrange the bosons and fermions, we would consider the result to be the same Adinkra.  But with $L_I$ matrices, since we have to number the bosons and number the fermions, the result may be a different set of $L_I$ matrices.  The intuition might be as follows: each boson and each fermion is numbered.  We might use the numbers $1, \ldots, d$ for the bosons, and re-use them $1, \ldots, d$ for the fermions.  Nothing particularly special is assumed about the relation between a boson and a fermion that is given the same number label.

In addition, we can reverse the sign on a vertex, and the result will reverse the dashing on all edges incident to that vertex.  The result gives an equivalent situation as far as supermultiplets are concerned, but it is still a different set of $L_I$ matrices.

These two issues can be captured by two signed permutation matrices $X$ and $Y$.  These can be used to modify the $L_I$ matrices to find $L_I'=XL_IY$ which also satisfy the above $\GR(d,N)$ algebra, and therefore provides another supermultiplet.  The classification of Adinkras assumed that sets of $L_I$ matrices that are related in this way were deemed to be equivalent.  But in this paper, this equivalence is not taken.

This equivalence is easy to defend, given that the isoscalar supermultiplets obtained from these $L_I$ matrices would only be modified by relabeling which field is which, and replacing fields with their negatives.  This seems like pretty much the same supermultiplet.  Why, then, should we care about the classification without this equivalence?

One answer is that the counting of the $L_I$ matrices is a straightforward mathematical question, independent of the physics, and as such, deserves an answer.

Alternately, we could note that there may be various structures that distinguish various bosons from each other or various fermions from each other.  There could be a symmetry group that relates some bosons together but not others.  This symmetry group may arise from a higher-dimensional supermultiplet that has been dimensionally reduced to one dimension, for instance.  Also, some of these bosons might be treated differently in some symmetry-breaking mechanism such as if these bosons were in an asymmetric external field. 
 The extreme case in these situations is for each boson and each fermion to be treated and labeled differently, which is the situation in the problem we have posed and addressed herein.

\section{Acknowledgments}
This work was partially supported by the National Science Foundation grant PHY-13515155.  S.J.G. acknowledges the generous support of the Roth Professorship and the very congenial and generous hospitality of the Dartmouth College physics department. This research was also supported in part the University of Maryland Center for String \& Particle Theory (CSPT).

\bibliographystyle{habbrv}
\bibliography{refs}

\end{document}